\documentclass[journal,10pt,twocolumn,english]{IEEEtran}
\usepackage{amsthm}
\usepackage{cite}
\usepackage{amsmath,amssymb,amsfonts}
\usepackage{mathtools}
\usepackage{algorithmic}
\usepackage{graphicx}
\usepackage{color}
\usepackage{epsfig}
\usepackage{rotating}
\usepackage{diagbox}
\usepackage{nicematrix}
\usepackage{tikz}
\usepackage{multirow}
\usepackage{subcaption}
\usepackage{arydshln}
\usepackage{hyperref}
\newcommand{\red}{\color{red}}

\newcommand{\insertxfig}[4]{
\begin{figure}[htbp]
\begin{center}
\leavevmode \centerline{\resizebox{#4\textwidth}{!}{\input
#1.pstex_t}}
\caption{#2} \label{#3}
\end{center}
\end{fontsize}}

\long\def\comment#1{}

\newcommand{\Ac}{{\mathcal A}}
\newcommand{\Bc}{{\mathcal B}}
\newcommand{\Cc}{{\mathcal C}}
\newcommand{\Dc}{{\mathcal D}}
\newcommand{\Ec}{{\mathcal E}}

\newcommand{\Gc}{{\mathcal G}}
\newcommand{\Hc}{{\mathcal H}}

\newcommand{\Kc}{{\mathcal K}}

\newcommand{\Mc}{{\mathcal M}}
\newcommand{\Nc}{{\mathcal N}}

\newcommand{\Qc}{{\mathcal Q}}
\newcommand{\Rc}{{\mathcal R}}
\newcommand{\Sc}{{\mathcal S}}
\newcommand{\Tc}{{\mathcal T}}

\newcommand{\csf}{{\mathsf c}}

\newcommand{\msf}{{\mathsf m}}
\newcommand{\nsf}{{\mathsf n}}

\newcommand{\qsf}{{\mathsf q}}
\newcommand{\rsf}{{\mathsf r}}
\newcommand{\ssf}{{\mathsf s}}

\newcommand{\Bsf}{{\mathsf B}}

\newcommand{\Ksf}{{\mathsf K}}
\newcommand{\Lsf}{{\mathsf L}}
\newcommand{\Msf}{{\mathsf M}}
\newcommand{\Nsf}{{\mathsf N}}

\newcommand{\Rsf}{{\mathsf R}}

\newcommand{\Tsf}{{\mathsf T}}
\newcommand{\Usf}{{\mathsf U}}

\newcommand{\Vsf}{{\mathsf V}}

\newtheorem{thm}{Theorem}%[section]
\newtheorem{rem}{Remark}

\newtheorem{constraint}{Constraint}

\providecommand{\definitionname}{Definition}

\begin{document}

\title{Hierarchical Secure Distributed Linearly Separable Computation  with Arbitrary Heterogeneous Data Assignment
}

\author{
\IEEEauthorblockN{%
Ziting Zhang,
Chenyi Sun,
Kai Wan,~\IEEEmembership{Member,~IEEE,}
Xiang Zhang,~\IEEEmembership{Member,~IEEE}
}
 \thanks{
A short version of this paper   was presented  at   the 2026 IEEE International Symposium on Information Theory (ISIT).
}
		\thanks{Z.~Zhang, C.~Sun, and K.~Wan are with the School of Electronic Information and Communications,
			Huazhong University of Science and Technology, 430074  Wuhan, China (e-mail: \{ziting\_zhang, chenyi\_sun, kai\_wan\}@hust.edu.cn).}
		\thanks{X. Zhang is with the Department of Electrical Engineering and Computer Science, Technical University of Berlin, 10623 Berlin, Germany (e-mail: xiang.zhang@tu-berlin.de).}

}

\maketitle

\begin{abstract}
This paper studies secure distributed linearly separable computation over a three-layer hierarchical network, where clustered users communicate with a central server through relays.  The server aims to recover $\Ksf_\csf$
linear combinations of $\Ksf$ intermediate outcomes, where each intermediate outcome is a separable function of one dataset. 
We consider a more general  setting with arbitrary heterogeneous data assignment across users, where `arbitrary' means that the data assignment is given in advance (which can be in any form) and `heterogeneous' means that the users may hold different numbers of datasets.  Under this assignment, each user computes the intermediate outcomes of its assigned datasets and sends masked messages to its associated relay. The relays subsequently process and forward the received messages to the server.  We impose two security constraints: (i) security against  server, requiring the server to learn only the desired task function without gaining any additional information about users’ inputs; and (ii) security against relays, ensuring  each relay learns nothing about users’ inputs. Moreover, the server or any relay may collude with a subset of users.
For $\Ksf_\csf=1$, the underlying computation reduces to distributed gradient coding. We propose a secure scheme tolerating user dropouts and user collusion, achieving the optimal two-layer communication rates in one regime and order-optimal communication rates within a factor of \(2\) in the other regime.
For $\Ksf_\csf>1$, we extend the proposed
construction to multi-dimensional linearly separable tasks under the no-dropout setting. %The extended scheme guarantees both server security and relay security against user collusion.
% In this process, while some users may drop out unpredictably, the server needs to correctly recover the task function from the surviving users. Moreover, the server or any relay may collude with a subset of users. We impose the following security constraints: (i) server security, requiring the server to learn only the desired task function without gaining any additional information about individual inputs; and (ii) relay security, ensuring that each relay learns nothing about users’ inputs. Under these requirements, we propose secure schemes for both $\Ksf_\csf=1$ and $\Ksf_\csf>1$.
%  Moreover, for $\Ksf_\csf=1$, we achieve the optimal two-layer communication loads.
\end{abstract}

\begin{IEEEkeywords}
Distributed linearly separable computation, information theoretic security, hierarchical network, arbitrary data assignment
\end{IEEEkeywords}

\section{Introduction}
Distributed computing underpins large-scale distributed learning and data analytics. In federated learning, users compute local gradients from their datasets, and a central server aggregates these local updates to obtain a global model update. To mitigate stragglers, gradient coding exploits redundancy in the data assignment, enabling the desired aggregate to be recovered from only a subset of user responses~\cite{gradiencoding,ye2018communication,cao2021adaptive}. This line of work focuses on a one-dimensional computation task, in which the server seeks to recover a single linear combination of the locally computed intermediate outcomes.

Beyond gradient aggregation, distributed computing systems may need to recover multiple linear combinations of the same intermediate outcomes, as in distributed linear transforms and coded matrix computations~\cite{dutta2016short,lee2017speeding,lee2017high,dutta2019optimal}. This broader class of tasks is captured by distributed linearly separable computation, in which the server aims to recover \(\Ksf_\csf\geq 1\)  linear combinations of 
\(\Ksf\) intermediate outcomes~\cite{wan2021distributed}.
Specifically, each intermediate outcome $W_k=f_k(D_k)$ is obtained by
applying a separable function $f_k(\cdot)$ to one dataset $D_k$. Thus, gradient aggregation corresponds to the special case \(\Ksf_\csf=1\), whereas \(\Ksf_\csf>1\) represents general multi-dimensional linearly separable tasks. Rather than focusing on how each local function is evaluated, this framework studies how the server can recover the desired linear combinations of these intermediate outcomes under data assignment and communication constraints. The associated computation-communication tradeoffs have been studied under homogeneous or structured data assignments \cite{wan2021distributed,wan2021tradeoff,huang2023fundamental}, while related coding-theoretic works formulate the joint design as a sparse matrix factorization problem \cite{khalesi2023multi,khalesi2024perfect,khalesi2025tessellated,namboodiri2025fundamental}.
Since intermediate outcomes may be sensitive, communication efficiency
alone is insufficient. Secure distributed linearly separable computation was studied in~\cite{wan2022secure} under the conventional star topology, where all users communicate directly with the server. Using random keys independent of the datasets, the server can recover the desired task without learning additional information about the inputs. The security requirement in this setting protects the users' inputs against the central server.

% Secure distributed linearly separable computation was studied
% in~\cite{wan2022secure}, where some random variables (as the keys) independent of the
% datasets are introduced to protect the intermediate outcomes.
% The security constraint requires that the server can recover the desired task function while learning no additional information about the inputs.
% It was shown that any linear-coding-based scheme for the non-secure
% distributed linearly separable computation problem can be transformed into a secure scheme without increasing the communication cost if the size of randomness is enough. In particular, for $\Ksf_\csf=1$ with the minimum computation cost, the optimal communication cost remains $\Nsf_\rsf$, and the minimum randomness size was further characterized under several data assignments.

\begin{table*}[t]
\centering
\caption{Comparison of the considered model with existing related works.}
\label{tab:model_comparison}
\renewcommand{\arraystretch}{1.18}
\setlength{\tabcolsep}{5pt}
\small
\begin{tabular}{
    p{0.10\textwidth}
    p{0.16\textwidth}
    p{0.20\textwidth}
    p{0.14\textwidth}
     p{0.24\textwidth}
}
\hline
\textbf{Work}
& \textbf{Topology}
& \textbf{Data assignment}
& \textbf{Dimension of computation task}
& \textbf{Security}
\\
\hline

\cite{gradiencoding,
ye2018communication,cao2021adaptive}
&Star network
&Non-fixed with redundancy
&\(\Ksf_\csf=1\)
&No security
\\
\cite{wan2021distributed,
wan2021tradeoff,huang2023fundamental,namboodiri2025fundamental}
&
Star network
&Non-fixed with redundancy
&
\(\Ksf_\csf\geq1\)
&No security
\\
\cite{wan2022secure}
&
Star network
&Non-fixed with redundancy
&
\(\Ksf_\csf\geq1\)
&Security against server
\\
\cite{zhang2024optimal,
zhang2025fundamental,xu2025hierarchical,
li2025hierarchical}
&
Hierarchical network
&
Fixed without redundancy
&
\(\Ksf_\csf=1\)
&Security against server and relays
\\
\cite{jahani2021optimal}
& Star network
& Fixed with redundancy
& \(\Ksf_\csf=1\)
&No security
\\
\cite{zhang2026distributed}
& Star network
& Fixed with redundancy
&  \(\Ksf_\csf\geq 1\)
&No security
\\

\cite{gholami2025hierarchicalgradientcodingoptimal}
& Hierarchical network
& Fixed with redundancy
&  \(\Ksf_\csf=1\)
&Security against relays
\\
\text{This work}
&
Hierarchical network
&
Fixed with redundancy
&
\(\Ksf_\csf\geq 1\)
&Security against server and relays
\\

\hline
\end{tabular}
\end{table*}

As the number of users grows, however, direct user-server communication creates bandwidth and connectivity
burdens. A hierarchical architecture groups users into clusters that communicate with the server through intermediate relays.  The introduction of intermediate relays changes not only the network topology but also the trust model. Hierarchical secure aggregation (HSA) has been studied under different user-relay association and collusion models~\cite{zhang2024optimal,zhang2025fundamental,xu2025hierarchical,li2025hierarchical}.
Since each relay receives and processes the messages transmitted by the users in its cluster, the security against server alone is no longer sufficient; the users' inputs must also be protected from the intermediate relays. This additional requirement, referred to as security against relays, distinguishes HSA from secure computation over a conventional star network.
% Existing secure distributed linearly separable computation schemes are typically designed for the conventional star topology, where all
% users communicate directly with a central server. As the number of users grows, however, such direct communication may lead to severe bandwidth bottlenecks and heterogeneous network conditions.
% This motivates hierarchical secure aggregation
% (HSA)~\cite{zhang2024optimal,zhang2025fundamental,
% xu2025hierarchical,li2025hierarchical}, where users communicate with the server through an intermediate layer of relays.
A three-layer hierarchical network was considered
in~\cite{zhang2024optimal}, consisting of one server, $\Usf$ relays,
and $\Usf\Vsf$ users under a symmetric architecture, where each relay
serves $\Vsf$ users. Two
information theoretic security constraints are imposed:
(i) security against server, requiring the server to learn only the sum of gradients without gaining any additional information about individual inputs; and (ii) security against relays, ensuring  each relay learns nothing about users’ inputs. These security
constraints must continue to hold even when the server or any relay
colludes with up to $\Tsf$ users.
The relay-to-server and user-to-relay communication rates, denoted by
$\Rsf_1$ and $\Rsf_2$, respectively, are defined as the maximum
normalized transmission loads over the corresponding links. The
capacity region was characterized as
$\{(\Rsf_1,\Rsf_2):\Rsf_1\geq 1, \Rsf_2\geq 1\}$
when $\Tsf<(\Usf-1)\Vsf$, whereas secure aggregation is infeasible
when $\Tsf\geq(\Usf-1)\Vsf$. 
Subsequent works extended HSA to more general association and security models. The work in~\cite{zhang2025fundamental} considered a cyclic user-relay association pattern, where each user is connected to $\Bsf$ relays in a cyclic wrap-around manner. It characterized the optimal communication cost and key rate region when $\Bsf\le \Ksf-1$, showing that multi-relay association can reduce both the relay-to-server communication load and the required key size. Another work studied HSA against relays and user collusion in homogeneous multi-association networks~\cite{xu2025hierarchical}, where each user is connected to $\nsf$ relays and each relay serves $\msf$ users. By establishing a connection between HSA and network function computation, communication-optimal schemes were obtained under certain collusion thresholds.  Heterogeneous security requirements were also investigated, where only prescribed subsets of users need to be protected against prescribed colluding sets~\cite{li2025hierarchical}. In this setting, the communication rates remain optimal, while the source key rate is determined by the interaction between the protected input sets and the colluding user sets. 

Existing HSA models are primarily motivated by federated learning, where each user holds a private dataset and the server seeks to recover the sum of these updates. This model does not fully capture more general distributed computing systems, 
since data placement is often determined by data generation, caching, or storage constraints before the coding scheme is designed, and users may hold irregular subsets of datasets. This motivates the study of arbitrary heterogeneous data assignment, where the assignment is fixed in advance and may have an arbitrary structure across users.
% Another important aspect of distributed computation is the data
% assignment across users. Most existing coded computing schemes rely
% either on coded data assignment or on structured uncoded assignments,
% such as cyclic and repetition-based assignments. Coded assignment
% requires centralized preprocessing of the datasets, while structured
% uncoded assignments fail to support arbitrary data assignment in
% heterogeneous environments.
Some works consider arbitrary heterogeneous data assignment under a conventional star topology. Gradient coding under arbitrary heterogeneous data assignment was studied in~\cite{jahani2021optimal}, where a universal coding scheme was proposed. In the presence of \(s\) stragglers, the minimum normalized communication cost was characterized as
$\frac{1}{r-s}$, where $r$ denotes the minimum replication factor
among all data partitions. Hence, the communication performance is
determined by the least-replicated data partition. 
By allowing the server to recover \(\Ksf_\csf\geq1\) linear combinations of the intermediate outcomes,  distributed linearly separable computation under arbitrary heterogeneous data assignment was further studied in~\cite{zhang2026distributed}. This work extended the setting to multi-dimensional tasks and characterized the fundamental tradeoff between communication cost and computable task dimension.
% {\red Building on this line of research, distributed linearly separable
% computation with arbitrary heterogeneous data assignment was further
% studied in~\cite{zhang2026distributed}. Different from the
% gradient coding setting with a one-dimensional task, this work considers
% multi-dimensional linearly separable tasks with \(\Ksf_\csf\geq 1\) and
% characterizes the fundamental tradeoff between the communication cost
% and the computable task dimension.
% }

More recently, this arbitrary heterogeneous data assignment model was extended to an asymmetric three-layer hierarchical network in~\cite{gholami2025hierarchicalgradientcodingoptimal}, where relay \(u\) serves a cluster of \(\Vsf_u\) users. For the one-dimensional aggregation task \(\Ksf_\csf=1\), the optimal two-layer communication region was characterized under homogeneous rate constraints. Specifically, let \(\Rsf_1\) denote the relay-to-server communication load and let \(\Rsf_2^{(u)}\) denote the user-to-relay communication load in cluster \(u\). In the absence of security constraints, the optimal region is
$\Big\{(\Rsf_1,\Rsf_2^{(u)}):\Rsf_1\geq \frac{1}{\rsf_1-\ssf_1}, \Rsf_2^{(u)}\geq \frac{1}{(\rsf_1-\ssf_1)(\rsf_2^{(u)}-\ssf^{(u)}_2)}\Big\}$, for each $u\in[\Usf]$,
where $\rsf_1$ and $\rsf_2^{(u)}$ denote the minimum
replication factors across relays and among the users in cluster \(u\),
respectively, \(\ssf_1\) denotes the maximum number of straggling
relays, and \(\ssf_2^{(u)}\) denotes the maximum number of
straggling users in cluster \(u\).
Its secure extension guarantees 
security against relays only under the restrictive conditions
$\ssf_1+1\leq\rsf_1\leq\Usf-1$ and
$\rsf_2^{(u)}=\ssf_2^{(u)}+1$, while preserving the same
communication rates.

However, as summarized in Table~\ref{tab:model_comparison}, the existing study of hierarchical networks with arbitrary heterogeneous data assignment is restricted to the one-dimensional aggregation task with \(\Ksf_\csf=1\). Moreover, its secure extension guarantees only security against relays under restrictive conditions and does not simultaneously ensure security against  server. Consequently, the fundamental limits of secure distributed linearly separable computation over such networks, where the server aims to recover \(\Ksf_\csf\geq1\) linear combinations, remain unexplored. In particular, it is unclear how the heterogeneous data availability at both the user and relay layers jointly affects the two-layer communication costs and the security against colluding users. Motivated by this gap, we investigate secure distributed linearly separable computation over an asymmetric three-layer hierarchical network with arbitrary heterogeneous data assignment.

\begin{figure}
  \centering
  \includegraphics[width=0.45\textwidth]{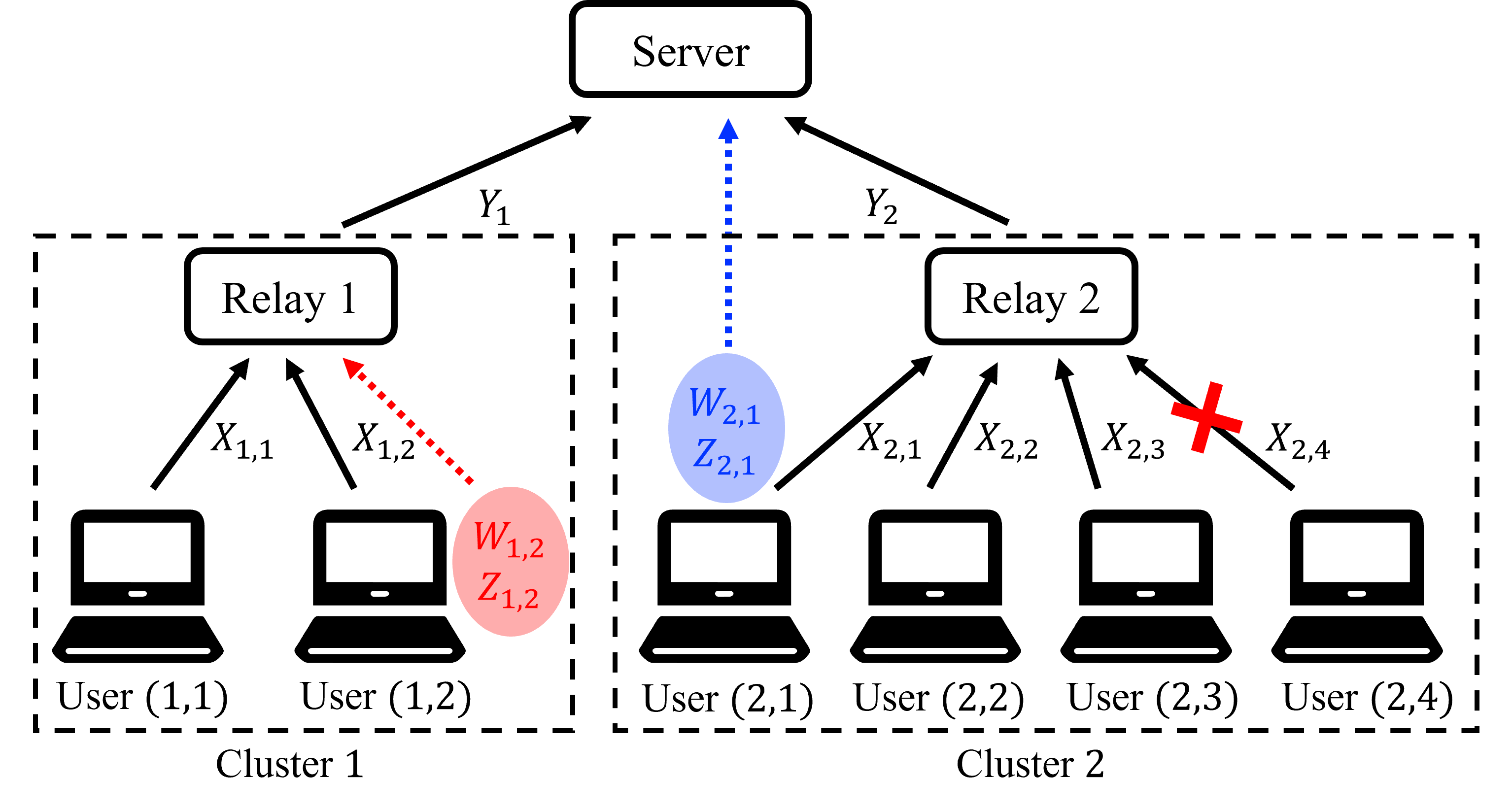}
  \caption{Hierarchical secure distributed linearly separable computation problem $(\Usf,\Vsf_1,\Vsf_2)=(2,2,4)$ where user $(2,4)$ is a straggler, relay 1 colludes with user $(1,2)$ and the server colludes with user $(2,1)$.}
  \label{fig: hierarchical secure aggregation}
\end{figure}

\paragraph*{Main Contributions}
Based on secure distributed linearly separable
computation~\cite{wan2022secure} and hierarchical gradient coding with arbitrary data assignment~\cite{gholami2025hierarchicalgradientcodingoptimal}, we formulate the information theoretic secure distributed linearly
separable computation over a three-layer hierarchical network, as illustrated in Fig.~\ref{fig: hierarchical secure aggregation}.
The network consists of a central server, multiple relays, and clusters of users with arbitrary heterogeneous data assignment. The server aims to recover
\(\Ksf_\csf\ge 1\) linear combinations of the intermediate
outcomes from surviving users, while guaranteeing security against server and relays with user collusion. We aim
to characterize the capacity region
\((\Rsf_1,\Rsf_2,\Rsf_Z)\).
%\footnote{Unlike existing hierarchical secure aggregation models, our model considers arbitrary heterogeneous data assignment.}
More precisely, our main contributions are as follows.
\begin{itemize}
    \item {\bf Secure gradient coding with user dropouts and
    collusion.}
    For the case $\Ksf_\csf=1$, where the desired task reduces to gradient aggregation, we consider up to
    $\ssf_2^{(u)}<\rsf_2^{(u)}$ user dropouts and at most $\Tsf_u$ colluding users in each cluster $u\in[\Usf]$.
    We propose a secure linear-coding-based scheme against  user dropouts and user collusion that achieves the optimal two-layer communication rates in one regime and order-optimal communication rates within a factor of \(2\) in the other regime. The main technical challenge lies in simultaneously satisfying the heterogeneous encodability constraints at both layers while preserving security against server and relays with user collusion. To address this challenge, we develop a security-aware nested task construction in which virtual demands and key-encrypted linear combinations are carefully designed. In particular, the virtual-demand and source-key dimensions are carefully chosen to compensate for the constraints imposed by unavailable datasets and to mask the additional information potentially exposed through user collusion.
    \item {\bf Multi-dimensional tasks under the no-dropout setting.}
    For the case $\Ksf_\csf>1$, we further extend the construction to
    multi-dimensional linearly separable tasks under arbitrary
    heterogeneous data assignment.
    By applying the heterogeneous construction in a layer-wise manner and treating the source keys as symbols available to all users,
    we obtain achievable bounds on the computable dimension and  guarantee security against server and relays with user
    collusion when $\Tsf_u\leq t_{2,u}$ for each $u\in[\Usf]$.
\end{itemize}
% of both relay-to-server layer and intra-cluster layer even in the presence of stragglers and collusion. 
%In contrast to \cite{gholami2025hierarchicalgradientcodingoptimal}, our scheme additionally considers collusion and ensures the security of both the relays and the server.   
    % providing robustness against a predetermined number of straggles in each cluster. 
    % % In intra-cluster layer, each user computes the partial gradient based on its dataset and sends masked gradient to the associated relay. In relay-to-server layer, each relay computes a message according to messages from its assigned users. 
    % Our scheme characterizes the optimal rate region, with the minimal communication rate in the intra-cluster layers and the relay-to-server layer achieving simultaneously.
    % \item We ensure the secure aggregation of users' gradients even in the presence of user collusion. Given the collusion threshold in each cluster, our proposed scheme guarantees both server security and relay security by introducing additional source keys.

\paragraph*{Paper Organization}
This paper is organized as follows.
Section~\ref{sec: system model} introduces the hierarchical secure distributed linearly separable
computation problem with arbitrary heterogeneous data assignment, where both user dropouts and user collusion are considered.
Section~\ref{sec: main results} presents the main results of this paper. Section~\ref{sec: achievability} describes the proposed secure aggregation scheme.
Building on this construction, Section~\ref{sec:heterogeneous-extension} extends the scheme
to multi-dimensional linearly separable tasks under the no-dropout setting. Finally,  Section~\ref{sec:conclusions} concludes the paper.

\paragraph*{Notation Convention} $[n]$ denotes the set $[1:n]$; $[a:b]$ represents the range $\{a,a+1,\dots,b\}$.
% ; $\ev_{n,i}$ represents the vertical $n$-dimensional unit vector whose $i$-th entry is $1$ and the remaining entries are $0$. 
%$\mathbf{A}^{-1}$ denotes the inverse of $\mathbf{A}$. 
 Let $\mathbf{0}_{m\times n}$ and $\mathbf{1}_{m\times n}$ represent the all-zero and all-one matrices of dimension $m\times n$, respectively. For a set $\Sc$, we denote the $i^{\text{th}}$ smallest element by $\Sc(i)$. 
We let $X_\Sc = \{X_i\}_{i\in \Sc}$ for any index set $\Sc$. $\mathbf{M}(\Sc,.)$ and $\mathbf{M}(.,\Sc)$ 
represent the sub-matrices of $\mathbf{M}$  composed of the
rows  with indices in $\Sc$ and the   columns  with indices in $\Sc$, respectively. The matrix $[\mathbf{A};\mathbf{B}]$ is expressed in a format similar to Matlab, equivalent to $\begin{bmatrix}
     \mathbf{A}\\ 
     \mathbf{B}
 \end{bmatrix} $. Calligraphic symbols represent sets. Vectors and matrices are denoted using bold symbols. System parameters are represented in sans-serif font. $|\cdot|$ defines the cardinality of a set. $\mathbb{F}_{\qsf}$ represents a finite field of order ${\qsf}$. For any logical statement $\mathcal{E}$, let $\mathbf{1}\{\mathcal{E}\}$ denote its indicator function, which equals $1$ if $\mathcal{E}$ is true and $0$ otherwise.

\section{System Model}\label{sec: system model}
We consider the coding-theory-based formulation of 
 secure distributed linearly separable computation in a hierarchical network with arbitrary data assignment,  as illustrated in Fig.~\ref{fig: hierarchical secure aggregation}. The global  dataset $\Dc$ is partitioned into $\Ksf$ non-overlapping and equal-length subsets, denoted by $\Dc=\{\Dc_1, \ldots, \Dc_\Ksf\}$, which are assigned to users in an arbitrary manner. For each dataset $\Dc_k$, where $k\in[\Ksf]$, a function
$f_k(\cdot)$ is applied to generate an intermediate outcome
$W_k=f_k(\Dc_k)$, where $W_k$ consists of $\Lsf$ symbols over a finite field $\mathbb F_\qsf$. Each intermediate outcome $W_k$, where $k\in[\Ksf]$, is further partitioned into $p$ non-overlapping and equal-length
pieces, denoted by $W_k=(W_{k,1},W_{k,2},\ldots,W_{k,p})$,
where each piece consists of $\Lsf/p$ symbols over
$\mathbb F_\qsf$. Define the collection of all intermediate pieces as
$\mathbf{W}=[W_{1,1};\ldots;W_{\Ksf,1};
 W_{1,2};\ldots;W_{\Ksf,p}]$.

The server aims to recover $p\Ksf_{\csf}$ linear combinations of these
pieces, given by
\begin{equation}
\label{eq: desired task}
f(D_1,\ldots, D_\Ksf)= \begin{bmatrix}
        F_1\\
        \vdots\\
        F_{p\Ksf_\csf}
\end{bmatrix}=\mathbf{F}\mathbf{W},
\end{equation}
where
$\mathbf F$ is the task coefficient
matrix with elements uniformly i.i.d.  over $\mathbb F_\qsf$. 
In particular, when $\Ksf_{\csf}=1$, the task consists of $p$
linear combinations of the pieces. By choosing each linear combination to aggregate the corresponding piece across all datasets, these \(p\) outputs collectively recover the gradient aggregation $\sum_{k\in[\Ksf]}W_k$. In this paper, we consider $\Ksf_\csf\ge1$.

% We consider the task function to be linearly separable, i.e., the function $f$ can be expressed as
%  $$f(D_1,\ldots, D_\Ksf)=g(f_1(D_1),\ldots,f_\Ksf(D_\Ksf))=\mathbf{F}[W_1;\ldots;W_\Ksf],$$
% where for each dataset $\Dc_k$ where $k\in[\Ksf]$ we apply an arbitrary
% function and compute an intermediate value $W_k$. 
The network consists of three layers: a central server, an intermediate layer of $\Usf$ relays, and distributed users at the bottom layer. The server is connected to all relays, while relay $u$ is connected to a disjoint cluster of $\Vsf_u$ users. Specifically, the $v^{\text{th}}$ user in cluster $u$  is labeled as user $(u,v) \in [\Usf]\times[\Vsf_u]$. All communication links are  orthogonal and noiseless.
Regarding the data assignment, the set of datasets assigned to user $(u,v)$ is denoted by $\Dc^{(u,v)}$, where $\Dc^{(u,v)}\neq \emptyset$. Accordingly, the collection of datasets assigned to cluster $u$ is $\Dc^{(u)} \triangleq\bigcup_{v\in[\Vsf_u]}\Dc^{(u,v)}$.
For each dataset $\Dc_k$ where $k\in[\Ksf]$, let $\Nc_k$ (resp. $\Nc_k^{(u)}$) represent the set of clusters (resp. the set of users in cluster $u$) that are assigned dataset $\Dc_k$. The cardinality of this set $\rsf_{1,k}\triangleq \left|\Nc_k\right|$ (resp. $\rsf_{2,k}^{(u)}\triangleq \Big|\Nc_k^{(u)}\Big|$) indicates the replication factor of the dataset $\Dc_k$ among clusters (resp. in cluster $u$). Define $\rsf_1\triangleq \min_{k\in[\Ksf]}{\rsf_{1,k}}$, $\rsf_2^{(u)}\triangleq \min_{k:\Dc_k\in\Dc^{(u)}}{\rsf_{2,k}^{(u)}},$ where $\rsf_1$ (resp. $\rsf_2^{(u)}$) represents the minimum replication factor among clusters (resp. in cluster $u$). 
% , and $\rsf_2^{(u)}$ represents the minimum replication factor in cluster $u$.

% Based on the assigned datasets $\Dc^{(u,v)}$, user $(u,v)$ computes $W_{u,v}$ as its input.
% To ensure secure aggregation, each user also stores an individual key $Z_{u,v}$, which is generated from the source key $Z_\Sigma$, i.e.,

Based on its assigned datasets $\Dc^{(u,v)}$, user $(u,v)$ computes the locally available intermediate outcomes $W_{u,v}=\{W_k:\Dc_k\in\Dc^{(u,v)}\}$. To guarantee secure computation, a source key \(Z_\Sigma\), independent of the intermediate outcomes, is generated before communication.  The source
key consists of a collection of mutually independent subkeys, where each
subkey contains $\Lsf/p$ symbols over $\mathbb{F}_\qsf$. Each user $(u,v)$ stores an individual key $Z_{u,v}$, which is generated
as $\Rsf_Z^{(u,v)}$ linear combinations of the subkeys in $Z_{\Sigma}$. 
The individual key rate is defined as the maximum normalized key size among all users,
\begin{align}
    \Rsf_Z
    &\triangleq
    \max_{u\in[\Usf],v\in[\Vsf_u]}
    \frac{\Rsf_Z^{(u,v)}}{p}.
    \label{eq:individual-key-rate}
\end{align}
% Accordingly,
% \begin{align}
%    &H\big(\left\{Z_{u,v}\right\}_{(u,v)\in[\Usf]\times[\Vsf_u]}\big|Z_\Sigma\big)=0,\label{eq: key generation}\\
%     &H(\left\{Z_{u,v}\right\}_{(u,v)\in[\Usf]\times[\Vsf_u]},\left\{W_{u,v}\right\}_{(u,v)\in[\Usf]\times[\Vsf_u]})\nonumber\\&
%     =H(\left\{Z_{u,v}\right\}_{(u,v)\in[\Usf]\times[\Vsf_u]})+H(\left\{W_{u,v}\right\}_{(u,v)\in[\Usf]\times[\Vsf_u]}).\label{the keys are independent}
% \end{align}

% The keys $\left\{Z_{u,v}\right\}_{(u,v)\in[\Usf]\times[\Vsf_u]}$ are independent of the inputs $\left\{W_{u,v}\right\}_{(u,v)\in[\Usf]\times[\Vsf_u]}$, 
% \begin{align}
% \label{the keys are independent}
%     &H(\left\{Z_{u,v}\right\}_{(u,v)\in[\Usf]\times[\Vsf_u]},\left\{W_{u,v}\right\}_{(u,v)\in[\Usf]\times[\Vsf_u]})\nonumber\\
%     =&H(\left\{Z_{u,v}\right\}_{(u,v)\in[\Usf]\times[\Vsf_u]})+H(\left\{W_{u,v}\right\}_{(u,v)\in[\Usf]\times[\Vsf_u]}).
% \end{align}

% The keys $\Zc \triangleq \left\{Z_{u,v}\right\}_{(u,v)\in[\Usf]\times[\Vsf_u]}$ are independent of the inputs $\Wc \triangleq \left\{W_{u,v}\right\}_{(u,v)\in[\Usf]\times[\Vsf_u]}$,
% \begin{equation}
%     H(\Zc,\Wc)=H(\Zc)+H(\Wc).
% \end{equation}
% i.e., 
% \begin{align}
% H\big(X_{u,v}|W_{u,v}, Z_{u,v}\big) = 0, \ \forall (u,v)\in[\Usf]\times[\Vsf_u].\label{eq: user message generation}
%  \end{align}
Over the first hop, user $(u,v)$ sends a masked input $X_{u,v}=(X_{u,v,i}:i\in[\Rsf_2^{(u,v)}])$  to its associated relay $u$, as a function of $W_{u,v}$ and $Z_{u,v}$.
Specifically, each user $(u,v)$ transmits $\Rsf_2^{(u,v)}$ linearly encoded messages, each constructed as
\begin{align}
X_{u,v,i} = \mathbf{E}_{u,v,i}\mathbf{W}+
\mathbf{G}_{u,v,i}Z_{u,v}\label{eq: user message generation},\ \forall i\in[\Rsf_2^{(u,v)}],
\end{align}
where $$\mathbf{E}_{u,v,i}=\{
e_{u,v,i,1,1},\ldots,e_{u,v,i,\Ksf,1},e_{u,v,i,1,2},\ldots,e_{u,v,i,\Ksf,p}\}
$$ and  $e_{u,v,i,k,l}$ for each $k\in[\Ksf]$ and $l\in[p]$ denotes the encoding coefficient associated with $W_{k,l}$ in the $i$-th encoded message. Thus, the message $X_{u,v}$ can be expressed as
\begin{align}
     \begin{bmatrix}
        X_{u,v,1}\\
        \vdots\\
        X_{u,v,\Rsf_2^{(u,v)}}
    \end{bmatrix}&=\begin{bmatrix}
       \mathbf{E}_{u,v,1}\\
         \vdots\\
   \mathbf{E}_{u,v,\Rsf_2^{(u,v)}}
    \end{bmatrix}\mathbf{W}+\begin{bmatrix}
       \mathbf{G}_{u,v,1}\\
         \vdots\\
   \mathbf{G}_{u,v,\Rsf_2^{(u,v)}}
\end{bmatrix}Z_{u,v}\nonumber\\&=\mathbf{E}_{u,v}\mathbf{W}+\mathbf{G}_{u,v}Z_{u,v}.
\end{align}
Since each user $(u,v)$ only encodes the datasets assigned to it, the coefficients must satisfy the encodability constraint 
\begin{align}
e_{u,v,i,k,l}=0,
\forall\, i\in[\Rsf_2^{(u,v)}],\;
k:\Dc_k\notin\Dc^{(u,v)},\;
l\in[p].
\end{align}

Thus, the user-to-relay communication rate is defined as the maximum normalized transmission load among all users,
\begin{align}
    \Rsf_2
    &\triangleq
   \max_{u\in [\Usf],v\in[\Vsf_u]} \frac{\Rsf_2^{(u,v)}}{p}.
    \label{eq:user-communication-rate}
\end{align}
%i.e., 
% \begin{align}
%     &H\big(Y_u \big|\left\{X_{u,v}\right\}_{v\in\Mc_u}\big)=0, \  \forall u\in[\Usf].\label{eq: relay message generation}
% \end{align}

Among $\Vsf_u$ users inside cluster $u$, we consider at most $\ssf^{(u)}_2=\Vsf_u-\Usf_u<\rsf_2^{(u)}$ stragglers.
Over the second hop, relay $u$ sends a message $Y_u=(Y_{u,j}:j\in[\Rsf_1^{(u)}])$ to the server, based on the messages received from the surviving users in $\Mc_u$, where $\Mc_u\subseteq\{u\}\times[\Vsf_u], |\Mc_u|\geq \Usf_u$.
Specifically, relay \(u\) transmits $\Rsf_1^{(u)}$ linearly encoded messages,   
\begin{align}
 Y_{u,j} = \mathbf A_{u,j} \begin{bmatrix}
        X_{\mathcal M_u(1)}\\
        \vdots\\
        X_{\mathcal M_u(|\mathcal M_u|)}
    \end{bmatrix} , \ \forall j\in[\Rsf_1^{(u)}].
\end{align}
Thus, the message $Y_{u}$ can be expressed as
\begin{align}
     \begin{bmatrix}
        Y_{u,1}\\
        \vdots\\
        Y_{u,\Rsf_1^{(u)}}
    \end{bmatrix}&=\begin{bmatrix}
       \mathbf{A}_{u,1}\\
         \vdots\\
   \mathbf{A}_{u,\Rsf_1^{(u)}}
    \end{bmatrix} \begin{bmatrix}
        X_{\mathcal M_u(1)}\\
        \vdots\\
        X_{\mathcal M_u(|\mathcal M_u|)}
    \end{bmatrix}= \mathbf{A}_{u}\begin{bmatrix}
        X_{\mathcal M_u(1)}\\
        \vdots\\
        X_{\mathcal M_u(|\mathcal M_u|)}
    \end{bmatrix}. \label{eq: relay message generation}
\end{align}
The relay-to-server communication rate is defined as the maximum normalized transmission load among all relays,
\begin{align}
    \Rsf_1
    &\triangleq
    \max_{u\in[\Usf]}  \frac{\Rsf_1^{(u)}}{p}.
    \label{eq:relay-communication-rate}
\end{align}

{\it Decodability.} All relays are   reliable, i.e., no relay straggling is considered. After receiving  the messages from all relays, the server
should be able to recover the desired task,
\begin{align}
    \begin{bmatrix}
        F_1\\
        \vdots\\
        F_{p\Ksf_\csf}
\end{bmatrix}=\mathbf{D}\begin{bmatrix}
        Y_1\\
         \vdots\\
        Y_\Usf
    \end{bmatrix}.
    \label{eq: decodability}
\end{align}

% i.e.,
% \begin{equation}\label{eq: decodability}
%     H\Bigg(\mathbf{F}[W_1;\ldots;W_\Ksf]\Big|\left\{Y_{u}\right\}_{u\in[\Usf]}\Bigg)=0.
% \end{equation}

For any colluding user set $\Tc_{u} \subseteq \{u\} \times[\Vsf_u]$  with $|\Tc_{u}|\leq \Tsf_{u}$ in each cluster $u\in[\Usf]$, denote the global colluding set  $\Tc \triangleq \bigcup_{u\in[\Usf]}\Tc_{u} $. The intermediate outcomes available to the colluding users are denoted by
$W_{\Tc}\triangleq \{W_k:\Dc_k\in \textstyle\bigcup_{(u,v)\in\Tc}\Dc^{(u,v)}\}$,
and let the colluding keys be denoted by
$Z_{\Tc}\triangleq (Z_{u,v}:(u,v)\in\Tc)$.
The security requirements ensure that, even if the server or any relay
colludes with up to $\Tsf_u$ users in each cluster
$u\in[\Usf]$, the following two security
constraints are satisfied.

% For any colluding set of users $\Tc\subseteq [\Usf]\times[\Vsf_u]$ with $|\Tc|\leq \Tsf$, let the colluding partial gradients be denoted by
% $W_{\Tc}\triangleq (W_k:\Dc_k\in \textstyle\bigcup_{(u,v)\in\Tc}\Dc^{(u,v)})$,
% and let the colluding keys be denoted by
% $Z_{\Tc}\triangleq (Z_{u,v}:(u,v)\in\Tc)$.
% Secure aggregation ensures that, even if the server or any relay  can respectively collude with up to $\Tsf$ users, the two security constraints are satisfied. 
% Secure aggregation guarantees that even if the server or any relay colludes with a subset of users, where the number of colluding users in each cluster $u$ is at most $\Tsf^{(u)}$
% , the following two security constraints are satisfied.\footnote{ Note that the server or each relay may collude with different sets of users; a relay can collude with users both in its cluster and across other clusters.}
% Secure aggregation guarantees that even if the server or any relay colludes with a subset of up to $\Tsf^{(u)}$ users in cluster $u$, the following two security constraints are satisfied
% the server can only compute the aggregated gradients $\sum_{k\in[\Ksf]}W_k$, without gaining any additional information about the users' inputs $\{W_k\}_{k\in[\Ksf]}$. Furthermore, each relay cannot infer any information about the users' inputs, even under such collusion.
% $\Tsf_{\isf\nsf}$

{\it Security against server.} 
For each $\Tc_i \subseteq \{i\} \times[\Vsf_i]$ with $|\Tc_i|\leq \Tsf_i$, where $i\in[\Usf]$, the server cannot obtain any information about the users' inputs beyond the desired task,
\begin{equation}\label{eq: server security}
    I\Big(\{Y_u\}_{u\in[\Usf]};\{W_k\}_{k\in[\Ksf]}\Big|\mathbf{F}\mathbf{W},\{W_{\Tc},Z_{\Tc}\}\Big)=0.
    % I\Big(\{Y_u\}_{u\in[\Usf]};\{W_k\}_{k\in[\Ksf]}\Big|\sum_{k\in[\Ksf]}W_k,\{W_{u,v},Z_{u,v}\}_{(u,v)\in\Tc}\Big)=0,
\end{equation}
% Recall that the notation $\{W_{\Tc},Z_{\Tc}\}$ represents the set $\{W_{u,v},Z_{u,v}\}_{(u,v)\in\Tc}$.

% \vspace{-0.1cm}

{\it Security against relays.} 
For each $\Tc_i \subseteq \{i\} \times[\Vsf_i]$ with $|\Tc_i|\leq \Tsf_i$, where $i\in[\Usf]$, each  relay  cannot infer any information about  the users’ inputs,
% , nothing about the partial gradients $\{W_k\}_{k\in[\Ksf]}$ can be inferred through its received messages, even if it colludes with any set of users $\Tc^{(i)}$ in cluster $i$, where $|\Tc^{(i)}|\leq \Tsf^{(i)}$ and $i\in [\Usf]$. We have
\begin{equation}\label{eq: relay security}
    I\Big(\{X_{u,v}\}_{v\in[\Vsf_u]};\{W_k\}_{k\in[\Ksf]}\Big|\{W_{\Tc},Z_{\Tc}\}\Big)=0, \forall u\in [\Usf].
    % I\Big(\{X_{u,v}\}_{v\in[\Vsf_u]};\{W_k\}_{k\in[\Ksf]}\Big|\{W_{i,j},Z_{i,j}\}_{(i,j)\in\Tc}\Big)=0.
\end{equation}
% Note that we consider a strong relay security in~\eqref{eq: relay security}, where stragglers in cluster $u$ are not actually unresponsive but are too slow in their transmissions; the relay may receive all messages from its associated users.

% The communication rates of the relay-to-server layer $\Rsf_1$ and the user-to-relay layer $\Rsf_2$ characterize the maximum number of symbols contained in each message $Y_u$ and $X_{u,v}$ per input
% symbol, respectively. The maximum individual key rate $\Rsf_Z$ characterizes the maximum number of symbols that each individual key $Z_{u,v}$ contains per input symbol, i.e.,
% \begin{align}
%     &\Rsf_1:= \max_{u\in [\Usf]} \frac{|Y_u|}{\Lsf},\ \Rsf_2:= \max_{u\in [\Usf],v\in[\Vsf_u]}\frac{|X_{u,v}|}{\Lsf},\nonumber\\&\Rsf_{Z}:=\max_{u\in [\Usf],v\in[\Vsf_u]}\frac{|Z_{u,v}|}{\Lsf}.\label{load}
% \end{align}

%{\red For \(\Ksf_\csf=1\), by assuming that $\Lsf$ is large enough, the symbols of each \(W_k\) can be grouped into blocks, with each block identified as one symbol over a sufficiently large extension field of \(\mathbb F_{\qsf}\).  Hence, without loss of generality, we can directly assume that $\qsf$ is large enough as in~\cite{zhao2022information}. For \(\Ksf_\csf>1\),  let \(\varepsilon \) denote the probability that the decodability constraint in~\eqref{eq: decodability} is not satisfied.  }

{\bf Objective.} 
A rate tuple $(\Rsf_1, \Rsf_2,\Rsf_{Z})$ is said to be achievable if there exists a secure linear coding scheme  with
user encoding matrices
$\{\mathbf{E}_{u,v},\mathbf{G}_{u,v}:(u,v)\in[\Usf]\times[\Vsf_u]\}$,
relay encoding matrices
$\{\mathbf{A}_{u}:u\in[\Usf]\}$,
and a decoding matrix $\mathbf{D}$, such that the security constraints
in~\eqref{eq: server security} and~\eqref{eq: relay security} are satisfied,
and the server can recover the desired task exactly for \(\Ksf_\csf=1\) and a vanishing error probability $\varepsilon  \stackrel{\qsf \to \infty}{\longrightarrow}  0$ for \(\Ksf_\csf>1\).\footnote{\label{foot:q large}For \(\Ksf_\csf=1\), by assuming that $\Lsf$ is large enough, the
symbols of each \(W_k\) can be grouped into blocks, with each block
identified as one symbol over a sufficiently large extension field of
\(\mathbb F_{\qsf}\).  Hence, without loss of generality, we can directly assume that $\qsf$ is large enough as in~\cite{zhao2022information}.
For \(\Ksf_\csf>1\), 
let \(\varepsilon \) denote the probability that the decodability
constraint in~\eqref{eq: decodability} is not satisfied.}
The objective is to characterize the capacity region (i.e., the closure of the set of all achievable rate tuples), denoted by $\Rc$.

% {\bf Existing converse and achievable bounds.} A converse bound for the hierarchical aggregation problem with arbitrary data assignment, considering user dropouts but without secure constraints, was proposed in~\cite{gholami2025hierarchicalgradientcodingoptimal}, which can be directly applied to our considered problem.
\begin{thm}[\cite{gholami2025hierarchicalgradientcodingoptimal}]\label{thm: related result}
For the distributed linearly separable computation problem in a hierarchical network
with arbitrary data assignment and $\Ksf_{\csf}=1$, consisting of one server, $\Usf$ relays, and $\Vsf_u$ users in cluster $u\in[\Usf]$  with at most $\ssf^{(u)}_2<\rsf_2^{(u)}$  stragglers, the minimum communication rates are characterized by
    \begin{equation}
        \mathcal{R}=\left\{
        \left(\Rsf_1,\Rsf_2^{(u)}\right):\Rsf_1\geq \frac{1}{\msf_1}, \Rsf_2^{(u)}\geq \frac{1}{\msf_1\msf_2^{(u)}},  \forall u\in[\Usf]
        \right\},
        \label{eq:capacity region}
    \end{equation}
    where $\msf_1 \triangleq \rsf_1$, $\msf_2^{(u)} \triangleq \rsf_2^{(u)}-\ssf^{(u)}_2$. In addition, when the security against relays is considered (i.e.,~\eqref{eq: relay security}), the optimal rate region remains as~\eqref{eq:capacity region} when $1\leq\rsf_1\leq\Usf-1$ and
$\rsf_2^{(u)}=\ssf_2^{(u)}+1$.
\end{thm}

\begin{rem}
 % The above result assumes that the transmission load is identical for all relays and that the transmission load is the same for all users within a cluster. However, in our model, we instead  define the transmission load of the relay-to-server layer $\Rsf_1$ as the maximum transmission load among all relays, and the transmission load of the user-to-relay layer $\Rsf_2$ as the maximum transmission load among all users, as described in~\eqref{load}.
 The above result is derived under a homogeneous rate constraint, where all relays have the same transmission load and all users within each cluster are subject to the same transmission load. For a given dataset, if the users in cluster \(u\) transmit fewer coded messages involving this dataset, the corresponding relay may accordingly transmit fewer key-encrypted coded combinations. To preserve decodability, the resulting reduction in information must then be compensated by increased transmissions from users in other clusters storing the same dataset, which may further increase the transmission loads of their associated relays. So, reducing the communication load in one cluster may necessarily increase that in others. 
 This inter-cluster user-load tradeoff motivates our definition of the worst-link user-to-relay communication load $\Rsf_2$  in~\eqref{eq:user-communication-rate}.
\end{rem}

\section{Main Results}
\label{sec: main results}
\subsection{Distributed Gradient Coding ($\Ksf_{\csf}=1$)}
\label{gradient coding}
Recall that $\rsf_1$ represents the
minimum replication factor among clusters, and $\rsf_2^{(u)}$ denotes the minimum replication factor within cluster $u$. 
Consider a cluster $u\in[\Usf]$ satisfying $\Tsf_u>\Vsf_u-\rsf_2^{(u)}$, since each dataset in $\Dc^{(u)}$ is assigned to at least $\rsf_2^{(u)}$ users, the complement of any colluding set $\Tc_u\subseteq\{u\}\times[\Vsf_u]$ with $|\Tc_u|=\Tsf_u$ contains fewer than $\rsf_2^{(u)}$ users. Hence, every dataset in $\Dc^{(u)}$ must be stored by at least one colluding user in $\Tc_u$. A stronger condition,
$\Tsf_u\geq\Vsf_u-\ssf_2^{(u)}$, implies that all surviving user messages in this cluster are available to the colluding users. This stronger condition leads to the second converse bound on the relay-to-server communication rate. 
We first establish converse bounds on the communication rates in the following theorem, whose proof is provided in Appendix~\ref{appendix:theorem for r_1=U}.
\begin{thm}
\label{thm: converse}
For the secure distributed linearly separable computation problem in a hierarchical network
with arbitrary data assignment and $\Ksf_{\csf}=1$, consisting of one server, $\Usf$ relays, and $\Vsf_u$ users in cluster $u\in[\Usf]$, with at most $\ssf^{(u)}_2<\rsf_2^{(u)}$ stragglers and at most $\Tsf_u$ colluding users, any achievable rate tuple should satisfy 
% \begin{align}
%     \Rsf_1 \ge \frac{1}{\msf_1},\Rsf_2 \ge \frac{1}{ \msf_1\nsf_2  }, \Rsf_Z\ge\frac{1}{ \msf_1\nsf_2  }.
% \end{align}
\begin{itemize}
    \item  $(\Rsf_1 \ge \frac{1}{\msf_1},\Rsf_2 \ge \frac{1}{ \msf_1\nsf_2  }, \Rsf_Z\ge\frac{1}{ \msf_1\nsf_2  })$ when %$1\leq\rsf_1\leq\Usf-1$ and
     $\sum_{u\in[\Usf]}
\mathbf{1}\{
\Tsf_u\ge\Vsf_u-\ssf_2^{(u)}\}
<\Usf-\rsf_1 $;
\item  $(\Rsf_1 \ge \frac{1}{\msf_1-1},\Rsf_2 \ge \frac{1}{ \msf_1\nsf_2  }, \Rsf_Z\ge\frac{1}{ \msf_1\nsf_2  })$ when  $\msf_1\geq2$ and  $\sum_{u\in[\Usf]}
\mathbf{1}\{
\Tsf_u\ge\Vsf_u-\ssf_2^{(u)}\}
\ge\Usf-\rsf_1$;
 \item   $\mathcal{R}=\emptyset$ when  $\msf_1=1$ and  $\sum_{u\in[\Usf]}
\mathbf{1}\{
\Tsf_u\ge\Vsf_u-\ssf_2^{(u)}\}
\ge\Usf-\rsf_1$,
\end{itemize}
where  $\msf_1 \triangleq \rsf_1$, $\msf_2^{(u)} \triangleq \rsf_2^{(u)}-\ssf_2^{(u)}$ and $\nsf_2 \triangleq\min_{u\in[\Usf]} \msf_2^{(u)}$.
\end{thm}
We explain the main theorem of this paper as follows.
\begin{itemize}
\item The threshold $\Usf-\rsf_1$ arises from the minimum cluster-level replication requirement. In particular, a cluster satisfying
$\Tsf_u\geq \Vsf_u-\ssf_2^{(u)}$
can be fully exposed through collusion, since all datasets assigned to this cluster and all surviving user messages can be available to the colluding users. Hence,
$\sum_{u\in[\Usf]}
\mathbf{1}\left\{
\Tsf_u\geq\Vsf_u-\ssf_2^{(u)}
\right\}$
counts the number of such exposed clusters. %If this number is smaller than $\Usf-\rsf_1$, more than $\rsf_1=\msf_1$ clusters remain effective, and the minimum replication constraint alone determines the converse bound in the first regime. Otherwise, at most $\rsf_1$ clusters remain effective. Since every non-colluding dataset must be assigned to at least $\rsf_1$ clusters and cannot be assigned to an exposed cluster, it must be assigned to all the remaining $\rsf_1=\msf_1$ clusters. Consequently, the security against relays further tightens the converse bound in the second regime.
    \item \emph{First regime.} Consider a dataset $\Dc_i$ whose cluster-level replication
    factor equals the minimum replication factor
    $\msf_1=\rsf_1$. Its contribution to the desired aggregation can only be conveyed through the $\msf_1$ relays storing it. Hence, at least one of these relays
    must transmit no fewer than $\Lsf/\msf_1$ symbols, which
    yields $\Rsf_1\geq 1/\msf_1$. For the bound on $\Rsf_2$, we continue to consider
    the same dataset $\Dc_i$. In a cluster $u$ storing $\Dc_i$,
    after at most $\ssf_2^{(u)}$ stragglers, only
    $\msf_2^{(u)}
    =\rsf_2^{(u)}-\ssf_2^{(u)}$ surviving users storing
    $\Dc_i$ are guaranteed to remain. These users must
    collectively provide the information required by the corresponding relay. Therefore, at least one of them
    must transmit no fewer than
$\Lsf/(\msf_1\msf_2^{(u)})$ symbols. Taking the worst case
    over all clusters gives
    $\Rsf_2\geq 1/(\msf_1\nsf_2)$.     Finally, to guarantee security against relays, the size of the key stored by each user must be at least as large as the maximum transmission size.
 Combining this requirement with the converse bound
    on $\Rsf_2$ yields
    $\Rsf_Z\geq 1/(\msf_1\nsf_2)$.  %The detailed proof is given in Appendix~\ref{proof for the first regime}.
\item  \emph{Second regime.} 
In this regime, only  the lower bound on $\Rsf_1$ is further tightened, whereas the converse bounds on $\Rsf_2$ and $\Rsf_Z$ remain unchanged. We first consider the case
$\sum_{u\in[\Usf]}
\mathbf{1}\left\{
\Tsf_u\ge\Vsf_u-\ssf_2^{(u)}
\right\}
=
\Usf-\rsf_1$. All datasets and keys in these $\Usf-\rsf_1$ clusters are fully exposed through collusion.
 % Since  $\Tsf_u\ge\Vsf_u-\ssf_2^{(u)}>\Vsf_u-\rsf_2^{(u)}$, all datasets in these $\Usf-\rsf_1$ clusters are  available to the colluding users, and all surviving user messages in these clusters  can also be exposed. 
 % Since  each dataset is assigned to at least $\rsf_1$ clusters, 
 % any non-colluding dataset  must be assigned to the remaining $\rsf_1$ clusters.
After excluding the $\Usf-\rsf_1$ clusters, the problem therefore reduces to a system with $\msf_1=\rsf_1$ effective clusters in which every non-colluding dataset is assigned to every effective cluster. By the security against relays, conditioned on the colluding side information available to any effective relay, its transmitted message must reveal no information about the non-colluding datasets. Hence, at the messages received by the server, the size of the keys should be at least equal to the maximum size of the transmission by each of the $\msf_1$ effective relays. 
Thus the total received dimension by the server is at most $\msf_1 \Rsf_1$, which should include  the task dimension (i.e., $1$) and the key dimension  (at least $\Rsf_1$), which leads to $\Rsf_1\geq 1/(\msf_1-1)$.
  When
$\sum_{u\in[\Usf]}
\mathbf{1}\left\{
\Tsf_u\ge\Vsf_u-\ssf_2^{(u)}
\right\}
>
\Usf-\rsf_1$, since the optimal rates are non-decreasing when the number of colluding users increases, the same converse bound holds.  %The detailed proof is given in Appendix~\ref{proof for the second regime}.
\item  \emph{Third regime.}  In this case, at most one effective relay remains. Consider a non-colluding dataset $\Dc_i$ that is assigned only to this effective relay. If this relay colludes with the other $\Usf-1$ clusters, it has access to its own received user messages together with all information available in the remaining clusters. In contrast, the server has access to the information from the other $\Usf-1$ clusters together with only the message transmitted by this relay, which is a function of the relay's received user messages. 
By the security against relays, the effective relay cannot recover any information about $W_i$ even colluding with the other $\Usf-1$ clusters and  knowing the remaining $\Ksf-1$ datasets. However, 
the server does not obtain more information than this `super' relay with the colluding information from the  other $\Usf-1$ clusters; thus the server 
%  Since the server has no more information about $\Dc_i$ than this effective relay, it 
 cannot recover  $W_i$ with the knowledge of   the remaining $\Ksf-1$ datasets either, which contradicts the decodability. Therefore, the security and decodability constraints cannot be simultaneously satisfied. 
%and secure distributed linearly separable computation is infeasible when $\msf_1=1$. %The detailed proof is given in Appendix~\ref{proof for the third regime}.

\end{itemize}

We then propose an achievable secure scheme by using a  hybrid design strategy.  
\begin{thm}
\label{thm: achievability}
Consider the secure distributed linearly separable computation problem in a hierarchical network
with arbitrary data assignment and $\Ksf_{\csf}=1$, consisting of one server, $\Usf$ relays, and $\Vsf_u$ users in cluster $u\in[\Usf]$, with at most $\ssf^{(u)}_2<\rsf_2^{(u)}$ stragglers and at most $\Tsf_u$ colluding users, 
\begin{itemize}
    \item  $(\Rsf_1 \ge \frac{1}{\msf_1},\Rsf_2 \ge \frac{1}{ \msf_1\nsf_2  }, \Rsf_Z\ge\frac{1}{ \msf_1\nsf_2  })$ is achievable when %$1\leq\rsf_1\leq\Usf-1$ and 
     $\sum_{u\in[\Usf]}
\mathbf{1}\{
\Tsf_u>\Vsf_u-\rsf_2^{(u)}\}
<\Usf-\rsf_1$;
\item  $(\Rsf_1 \ge \frac{1}{\msf_1-1},\Rsf_2 \ge \frac{1}{ (\msf_1-1)\nsf_2  }, \Rsf_Z\ge\frac{1}{ (\msf_1-1)\nsf_2  })$ is achievable when $\msf_1\geq2$ and $\sum_{u\in[\Usf]}
\mathbf{1}\{
\Tsf_u>\Vsf_u-\rsf_2^{(u)}\}\ge\Usf-\rsf_1$,
% \item   $(\Rsf_1 \ge \frac{1}{\msf_1},\Rsf_2 \ge \frac{1}{ \msf_1\nsf_2  }, \Rsf_Z\ge\frac{1}{ \msf_1\nsf_2  })$ when  $\sum_{u\in[\Usf]}
% \mathbf{1}\{
% \Tsf_u>\Vsf_u-\rsf_2^{(u)}\}
% >\Usf-\rsf_1$,
\end{itemize}
where  $\msf_1 \triangleq \rsf_1$, $\msf_2^{(u)} \triangleq \rsf_2^{(u)}-\ssf_2^{(u)}$ and $\nsf_2 \triangleq\min_{u\in[\Usf]} \msf_2^{(u)}$.
\end{thm}

We summarize the key idea of the  scheme as follows.
\begin{itemize}
    \item A hybrid design strategy is employed to construct the computational tasks for both relays and users. Specifically, in the relay-to-server layer, the aggregation problem is reduced to a secure gradient coding formulation over a star network, where relays are treated as ``special users'' and transmit coded messages protected by the keys. Once the tasks for each relay are assigned, a  gradient coding formulation  over a star network is applied in each cluster, which also treats the keys in the message transmitted by  the relay as the tasks in this gradient coding problem.
    % \item The source key rate by the proposed scheme is given  the first term corresponds to relay security, and the second term accounts for server security.
% The first term corresponds to relay security: for each relay $i$, except for the encrypted messages from $\Vsf_i - \ssf^{(i)}_{2}$ surviving users, $\sum_{j\in[\Usf]\setminus\{i\}}{\Tsf^{(j)}\Rsf_{2}^{(j)}}$ additional independent keys are required to address user collusion. The second term accounts for server security: the server receives encrypted messages from $\Usf$ relays, and obtains key information from $\Tsf^{(i)}$ colluding users in each cluster $i$. However, it can only recover the gradient task, and the source key rate is bounded by~\cite{zhao2023secure}.
 \item The proposed scheme guarantees the security against relays and  server.
% For the security against server, 
The polynomial code-based scheme in~\cite{gholami2025hierarchicalgradientcodingoptimal} is unable to guarantee server-side security. This is because,  from the polynomial recovered by the server, the number of inputs which zero-forces the keys is larger than the computational task, and thus information about the gradients beyond their sum is leaked. In contrast, we design the secure aggregation scheme based on linear space, ensuring that the key dimension in the linear space observed by the server is exactly sufficient for masking the additional information.
\end{itemize}

Comparing the proposed converse and achievable bounds, we  obtain the following optimality result. 
\begin{thm}
\label{thm: capacity}
For the secure distributed linearly separable computation problem in a hierarchical network with arbitrary data assignment and $\Ksf_{\csf}=1$, consisting of one server, $\Usf$ relays, and $\Vsf_u$ users in cluster $u\in[\Usf]$, with at most $\ssf_2^{(u)}<\rsf_2^{(u)}$ stragglers and at most $\Tsf_u$ colluding users,  
\begin{itemize}
    \item  $(\Rsf_1 \ge \frac{1}{\msf_1},\Rsf_2 \ge \frac{1}{ \msf_1\nsf_2  }, \Rsf_Z\ge\frac{1}{ \msf_1\nsf_2  })$ is optimal when 
    %$1\leq\rsf_1\leq\Usf-1$ and 
    $\sum_{u\in[\Usf]}
\mathbf{1}\{
\Tsf_u>\Vsf_u-\rsf_2^{(u)}\}
<\Usf-\rsf_1$;
\item  $(\Rsf_1 \ge \frac{1}{\msf_1-1},\Rsf_2 \ge \frac{1}{ (\msf_1-1)\nsf_2  }, \Rsf_Z\ge\frac{1}{ (\msf_1-1)\nsf_2  })$ is order-optimal within  a factor of  $2$ when $\msf_1\geq2$ and $\sum_{u\in[\Usf]}
\mathbf{1}\{
\Tsf_u>\Vsf_u-\rsf_2^{(u)}\}
\ge\Usf-\rsf_1$,
\end{itemize}
where  $\msf_1 \triangleq \rsf_1$, $\msf_2^{(u)} \triangleq \rsf_2^{(u)}-\ssf_2^{(u)}$ and $\nsf_2 \triangleq\min_{u\in[\Usf]} \msf_2^{(u)}$.
\end{thm}

\subsection{Multi-Dimensional Linearly Separable Tasks ($\Ksf_{\csf}\geq 1$)}
\label{Multi-Dimensional Linearly Separable Tasks}
We next consider the case $\Ksf_\csf\geq1$ under the no-dropout
setting, where all users and all relays remain active.\footnote{
The main challenge of considering $\Ksf_\csf\geq1$ with
user dropouts is that the impact of stragglers depends on
the heterogeneous data assignment. Different datasets may
be replicated among different user subsets, and thus the same set of
stragglers can result in different losses of available information for
different datasets. Unlike the case of $\Ksf_\csf=1$, where only one aggregate
function needs to be recovered, recovering multiple independent linear
combinations requires preserving sufficient dimensions for
all datasets for all possible
stragglers. 
}

For each cluster $u\in[\Usf]$, let $\Dc^{(u)}$ denote the set of datasets available in cluster $u$. We order the datasets in $\Dc^{(u)}$ according to their global indices and denote the $k$-th dataset in this ordered list by $\Dc^{(u)}(k)$, where $k\in[|\Dc^{(u)}|]$.
We define the local data assignment matrix $\mathbf A_{2,u}$ with dimension $\Vsf_u\times|\Dc^{(u)}|$,
\begin{align}
\mathbf A_{2,u}(v,k)
=
\begin{cases}
*, & \Dc^{(u)}(k)\in\Dc^{(u,v)},\\
0, & \Dc^{(u)}(k)\notin\Dc^{(u,v)}.
\end{cases}
\end{align}
Here, $\mathbf A_{2,u}$ describes the data assignment in cluster $u$. The keys are not included because they are available to all users.

For a given user-to-relay communication cost $\Rsf_{2}$, define
\begin{align*}
\mathcal Z_{2,u}
\triangleq
\{
(\mathcal G,\mathcal Q)&|\Gc\subseteq[\Vsf_u],\Qc\subseteq[|\Dc^{(u)}|],\nonumber\\&
\mathbf A_{2,u}(\mathcal G,\mathcal Q)=\mathbf 0,
\Rsf_{2}|\mathcal G|+|\mathcal Q|>\Rsf_{2}\Vsf_u
\}.
\end{align*}
Let
$
\mathcal G_{2,u}^\prime
\triangleq
\bigcup_{(\mathcal G,\mathcal Q)\in\mathcal Z_{2,u}}
\mathcal G. 
$
Let $t_{2,u}$ be the maximum number of zeros in each column of $\mathbf A_{2,u}(\mathcal G_{2,u}^\prime,.)$, i.e.,
$t_{2,u}
\triangleq
\max_{k\in[|\Dc^{(u)}|]}
\left|
\left\{
v\in\mathcal G_{2,u}^\prime:
\mathbf A_{2,u}(v,k)=0
\right\}
\right|$.
Then, by directly applying the heterogeneous computation scheme in~\cite{zhang2026distributed}, relay $u$ can recover
\begin{align}
\Ksf_{\csf}^{(u)}
\triangleq
\min\left\{
\Rsf_{2}(\Vsf_u-t_{2,u}),\ |\Dc^{(u)}|
\right\}
\label{eq:local_Kcu}
\end{align}
linear combinations of the data pieces available in cluster $u$ and keys.
Since different clusters may have different local data assignment matrices, the values
$
\Ksf_{\csf}^{(1)},\Ksf_{\csf}^{(2)},\ldots,\Ksf_{\csf}^{(\Usf)}
$
may also be different. Hence, we choose
\begin{align}
\Rsf_1
\triangleq
\min_{u\in[\Usf]}\Ksf_{\csf}^{(u)}, 
\label{eq:C1_min_Kcu}
\end{align}
which guarantees that each relay $u\in[\Usf]$ can locally recover at least $\Rsf_1$ linear combinations of the datasets $\{W_k:\Dc_k\in\Dc^{(u)}\}$ and keys, with coefficients generated uniformly i.i.d. over $\mathbb F_\qsf$.

We now apply the heterogeneous computation scheme in~\cite{zhang2026distributed} again across relays. We define the cluster-level data assignment matrix $\mathbf A_1$ with dimension $\Usf\times\Ksf$. Specifically,
\begin{align}
\mathbf A_1(u,k)
=
\begin{cases}
*, & \Dc_k\in\Dc^{(u)},\\
0, & \Dc_k\notin\Dc^{(u)}.
\end{cases}
\end{align}
Again, $\mathbf A_1$ describes the data assignment across relays. The keys are available to all relays.

For a given relay-to-server communication cost $\Rsf_1$, define
\begin{align*}
\mathcal Z_1
\triangleq
\{
(\mathcal G,\mathcal Q)&|\Gc\subseteq[\Usf],\Qc\subseteq[\Ksf],\nonumber\\&
\mathbf A_1(\mathcal G,\mathcal Q)=\mathbf 0,
\Rsf_1|\mathcal G|+|\mathcal Q|>\Rsf_1\Usf
\}.
\end{align*}
Let
$
\mathcal G_1^\prime
\triangleq
\bigcup_{(\mathcal G,\mathcal Q)\in\mathcal Z_1}
\mathcal G
$, and $t_1$ be the maximum number of zeros in each column of $\mathbf A_1(\mathcal G_1^\prime,.)$, i.e.,
$t_1
\triangleq
\max_{k\in[\Ksf]}
\left|
\left\{
u\in\mathcal G_1^\prime:
\mathbf A_1(u,k)=0
\right\}
\right|$.
The server can recover
\begin{align}
\Ksf_\csf^{\mathrm{srv}}
\triangleq
\min\left\{
\Rsf_1(\Usf-t_1),\ \Ksf
\right\}
\label{eq:server_Kc}
\end{align}
linear combinations of $\{W_1,\ldots,W_\Ksf\}$ and keys with coefficients uniformly i.i.d. over $\mathbb F_\qsf$. Letting the coefficients of the keys be zero, the server can recover the $\Ksf_\csf^{\mathrm{srv}}$-dimensional task.

\begin{thm}
\label{thm: multi-dimension}
Consider the secure distributed linearly separable computation problem in a hierarchical network
with arbitrary data assignment and $\Ksf_{\csf}>1$, consisting of one server, $\Usf$ relays, and $\Vsf_u$ users in cluster $u\in[\Usf]$, with at most $\Tsf_u\leq t_{2,u}$ colluding users.
For a given user-to-relay communication cost $\Rsf_2$,
the computable dimension in cluster $u$ is
$
\Ksf_{\csf}^{(u)}
\triangleq
\min\left\{
\Rsf_2\bigl(\Vsf_u-t_{2,u}\bigr),
\left|\Dc^{(u)}\right|
\right\},
$
and the relay-to-server communication cost is chosen as
$\Rsf_1
\triangleq
\min_{u\in[\Usf]}\Ksf_{\csf}^{(u)}$.
Then, the achievable task dimension at the server is
\begin{equation}
\Ksf_\csf\leq\Ksf_\csf^{\mathrm{srv}}
=
\min\left\{
\Rsf_1(\Usf-t_1),
\Ksf
\right\}.
\end{equation}
\end{thm}
The proposed  scheme for $\Ksf_{\csf}>1$ is presented in
Section~\ref{sec:heterogeneous-extension}.

\section{Proof of Theorem~\ref{thm: achievability}: Achievable Scheme}\label{sec: achievability}
In this section, we present the proposed secure aggregation scheme that achieves the rate region in Theorem~\ref{thm: achievability} for  the first regime (i.e., $\sum_{u\in[\Usf]}
\mathbf{1}\{
\Tsf_u>\Vsf_u-\rsf_2^{(u)}\}
<\Usf-\rsf_1$).  %is presented in Section~\ref{sec: achievability}.  
The construction for the second regime (i.e., $\msf_1\geq2$ and  $\sum_{u\in[\Usf]}
\mathbf{1}\{
\Tsf_u>\Vsf_u-\rsf_2^{(u)}\}
\ge\Usf-\rsf_1$) follows a similar coding principle, and the main difference will be summarized in Remark~\ref{rem:ClargerK}.
We begin with an illustrative example to highlight the main ideas of the construction.
\subsection{A Motivating Example}
\label{example}

% Consider an aggregation server connected to $\Usf = 2$ relays, where relay 1 and relay 2 are associated with $\Vsf_1 = 2$ and $\Vsf_2 = 4$ users, respectively.
% The training dataset is partitioned into $\Ksf = 6$ datasets, i.e., $\Dc =\{\Dc_1, \ldots, \Dc_6\}$, where each partial gradient $W_k$ is computed from the corresponding dataset $\Dc_k$. Consider an arbitrary data assignment as shown in Table~\ref{tab: dataset assignment}, with $\rsf_1 = 1$, $\rsf_{2}^{(1)} = 1$, and $\rsf_{2}^{(2)} = 3$.

Consider an aggregation server connected to $\Usf=2$ relays, with $\Vsf_1=2$ and $\Vsf_2=3$ users in the two clusters. The training dataset is partitioned into $\Ksf=6$ non-overlapping and equal-length datasets $\Dc=\{\Dc_1,\ldots,\Dc_6\}$, where each partial gradient $W_k$ is computed from $\Dc_k$. The datasets are assigned to users in an arbitrary manner, as shown in Table~\ref{tab:data-assignment}.  For this assignment, we have $\rsf_1=1$, $\rsf_2^{(1)}=1$, and $\rsf_2^{(2)}=2$. Assume that the numbers of stragglers in the two clusters are $\ssf_2^{(1)}=0$ and $\ssf_2^{(2)}=1$. Hence, $\msf_1=\rsf_1=1$ and $\nsf_2=\min_{u\in[2]}\msf_2^{(u)}=\min_{u\in[2]}(\rsf_2^{(u)}-\ssf_2^{(u)})=1$. In this example, we assume that $\qsf$ is a sufficiently large prime for the sake of illustration, although this assumption is not required in the general scheme. In addition, assume that $\Tsf_1=\Tsf_2=1$. Introduce three independent keys (each with $\Lsf$ elements uniformly i.i.d. over $\mathbb{F}_{\qsf}$), denoted by $Z_\Sigma=\{N_1,N_2,N_3\}$. 
\begin{table}
\caption{Data assignment}
\label{tab:data-assignment}
\centering
\renewcommand{\arraystretch}{1.1}
\begin{tabular}{c c c c c c c c}
\hline
\multirow{2}{*}{relay} & \multirow{2}{*}{user}
& \multicolumn{6}{c}{dataset} \\
\cline{3-8}
& & $\Dc_1$ & $\Dc_2$ & $\Dc_3$ & $\Dc_4$ & $\Dc_5$ & $\Dc_6$ \\
\hline
\multirow{2}{*}{1}
& $(1,1)$ & $\checkmark$ &              & $\checkmark$ & $\checkmark$ &              & $\checkmark$ \\
& $(1,2)$ & $\checkmark$ & $\checkmark$ &              &              &              & $\checkmark$ \\
\hline
\multirow{3}{*}{2}
& $(2,1)$ &              & & $\checkmark$ &              & $\checkmark$ &              \\
& $(2,2)$ &              & & $\checkmark$ & $\checkmark$ & $\checkmark$ &              \\
& $(2,3)$ &              & & $\checkmark$ & $\checkmark$ &              &              \\
\hline
\end{tabular}
\end{table}

{\bf Relay-to-Server layer.} Based on Theorem~\ref{thm: achievability}, we have $\Rsf_1 \geq \frac{1}{\msf_1}=1$. Each relay transmits one linear combination of  $\{W_1,\ldots,W_6,N_1,\ldots,N_3\}$ to the server. The server   recovers
\begin{align}\label{eq:eg-sub-tasks}
    \left[
    \begin{NiceArray}{c|[tikz={dash pattern=on 4pt off 3pt}]c}
        (\mathbf{F}_{1})_{2\times6}  & (\mathbf{B}_{1})_{2\times3}
    \end{NiceArray}
    \right] \negmedspace \begin{bmatrix}
        \mathbf{W}\\
        \mathbf{N}
    \end{bmatrix} \negmedspace\negmedspace
    = \negmedspace\negmedspace
    \left[
    \begin{array}{cc}
        \mathbf{F}_{1\times6}& \mathbf{0}_{1\times3} \\
        \mathbf{V}_{1\times6} & \mathbf{B}_{1\times3}
    \end{array}
    \right]  \negmedspace 
    \begin{bmatrix}
        W_1\\
        \vdots\\
        W_6\\
        N_1\\
        \vdots\\
        N_3
    \end{bmatrix};
\end{align}
we define  $\mathbf{W}^{\prime}\triangleq \begin{bmatrix}
        \mathbf{W}\\
        \mathbf{N}
    \end{bmatrix}$. The first row corresponds to the desired aggregation task, and hence $\mathbf{F}=[1\ 1\ 1\ 1\ 1\ 1]$, while the remaining  row is the virtual demand to be designed. 

Define the  matrix
$
    \mathbf{S}_1
    \triangleq
    \begin{bmatrix}
        \mathbf{S}^{(1)}_1\\
        \mathbf{S}^{(2)}_1
    \end{bmatrix}=\begin{bmatrix}
        3 & 1 \\
        2 & 1 
    \end{bmatrix},
$
whose elements are uniformly i.i.d. over $\mathbb{F}_{\qsf}$.
Then each relay $u\in[2]$ transmits
$
Y_u
=
\mathbf{S}^{(u)}_1
\left[
\begin{NiceArray}{c|[tikz={dash pattern=on 4pt off 3pt}]c}
\mathbf{F}_1 & \mathbf{B}_1
\end{NiceArray}
\right]
\mathbf{W}^\prime.
$
Since relay $u$ can only form linear combinations of the datasets available in cluster $u$, the coefficients corresponding to the datasets that are not in $\Dc^{(u)}$ must be `0'. Therefore, we design the matrix $\mathbf{V}$ according to the encodability constraints. For example, for the first column of $\mathbf{F}_1$, since $\Dc_1$ is not available at relay 2, we require $\mathbf{S}^{(2)}_1\mathbf{F}_1(.,\{1\})=0$. Thus, we can solve $\mathbf{F}_1(.,\{1\})=[1,-2]^\Tsf$. The remaining columns of $\mathbf{F}_1$ can be obtained in the same manner, and the detailed construction  is
\begin{align}
    \mathbf{F}_1=\begin{bmatrix}
        1&1&1&1&1&1\\
         -2&-2&3&1&-3&-2
    \end{bmatrix}.
\end{align}
Then we further choose the top-layer key coefficient matrix
$\mathbf{B}=\begin{bmatrix}
1 & 1 & 2 
\end{bmatrix}$,
whose elements are uniformly i.i.d. over $\mathbb{F}_\qsf$. Thus, we complete the design of the matrix in~\eqref{eq:eg-sub-tasks}
\begin{align}
 \left[
     \begin{NiceArray}{c|[tikz={dash pattern=on 4pt off 3pt}]c}
        \mathbf{F}_{1}  & \mathbf{B}_{1}
    \end{NiceArray}
    \right] \negmedspace\negmedspace
    = \negmedspace\negmedspace
    \left[
    \begin{NiceArray}{cccccc|[tikz={dash pattern=on 4pt off 3pt}]ccc}
        1&1&1&1&1&1&0&0&0\\
         -2&-2&3&1&-3&-2&1 & 1 & 2 
    \end{NiceArray}
    \right].
\end{align}
Consequently, the transmissions of the two relays  are 
\begin{subequations}
\begin{alignat}{2}
&Y_1 \negmedspace\negmedspace
= \negmedspace\negmedspace W_1 \negmedspace+ \negmedspace W_2+6W_3+4W_4+W_6+N_1+N_2+2N_3,\\
&Y_2
= 5W_3+3W_4-W_5+N_1+N_2+2N_3.
\end{alignat}
\label{exp:relay transmission}
\end{subequations}
\emph{Decodability.}
After receiving the messages from all relays, the server can recover
$
    \left[
    \begin{NiceArray}{c|[tikz={dash pattern=on 4pt off 3pt}]c}
        \mathbf{F}_1 & \mathbf{B}_1
    \end{NiceArray}
    \right]
    \mathbf{W}^\prime
    =
    \mathbf{S}_1^{-1}
    \begin{bmatrix}
        Y_1\\
        Y_2
    \end{bmatrix},
$
where the first row of the recovered vector is exactly the desired aggregation result $\sum_{k\in[6]}W_k$. Therefore, the server can decode the aggregation result from the relay transmissions.

{\bf User-to-Relay layer.}
 We assign the key-encrypted linear combinations in~\eqref{exp:relay transmission} as local computation tasks  to the clusters.  By $\Rsf_2\geq \frac{1}{\msf_1\nsf_2}=1$, each user transmits one linear combination to its associated relay.
 
For cluster $1$, the assigned relay-layer computation is
\begin{align*}
    \mathbf{S}^{(1)}_1 \left[
    \begin{NiceArray}{c|[tikz={dash pattern=on 4pt off 3pt}]c}
        \mathbf{F}_{1}  & \mathbf{B}_{1}
    \end{NiceArray}
    \right]\mathbf{W}^{\prime}=
   \begin{bmatrix}
   1 & 1 & 6 & 4 & 0 & 1 & 1 & 1 & 2
   \end{bmatrix}\mathbf{W}^{\prime}.
\end{align*}
To enable relay $1$ to recover this linear combination while preserving security, we let relay $1$ recover two key-encrypted linear combinations, i.e., 
\begin{align}
\left[
\begin{NiceArray}{c|[tikz={dash pattern=on 4pt off 3pt}]c}
\mathbf{F}_{2}^{(1)}  & \mathbf{B}_{2}^{(1)} 
\end{NiceArray}
\right]\mathbf{W}^\prime=\left[
\begin{array}{cc}
 \mathbf{S}^{(1)}_1\mathbf{F}_1 &  \mathbf{S}^{(1)}_1 \mathbf{B}_1\\
\mathbf{V}_{2,1} & \mathbf{B}_{2,1}
\end{array}
\right]
  \mathbf{W}^\prime,
\end{align}
where \(\mathbf{V}_{2,1}\) is a virtual demand coefficient matrix designed according to the users' encodability constraints, and \(\mathbf{B}_{2,1}\) denotes the lower-layer key coefficient matrix that is to be designed. 
We generate 
$
     \mathbf{S}_{2,1} \triangleq
    \begin{bmatrix}
        \mathbf{S}^{(1,1)}_{2,1}\\
        \mathbf{S}^{(1,2)}_{2,1}
    \end{bmatrix}=\begin{bmatrix}
        4 &-1  \\
        1 & -1  
    \end{bmatrix},
$
where the elements are uniformly i.i.d. over $\mathbb{F}_{\qsf}$.
Then each user $(1,v)$, where $v\in[2]$, transmits
$
X_{1,v}
=
\mathbf{S}^{(1,v)}_{2,1}
\left[
\begin{NiceArray}{c|[tikz={dash pattern=on 4pt off 3pt}]c}
\mathbf{F}_{2}^{(1)}   & \mathbf{B}_{2}^{(1)} 
\end{NiceArray}
\right]\mathbf{W}^\prime,
$
where the coefficients corresponding to the datasets that are not in $\Dc^{(1,v)}$ must be `0'.  Therefore, \(\mathbf{V}_{2,1}\) is chosen so that the encodability constraints of the two users are satisfied. Specifically, for the first column of $\mathbf{F}_{2}^{(1)}$, since $\Dc_1$ is available at users $(1,1)$ and $(1,2)$, we randomly select $\mathbf{V}_{2,1}(.,\{1\})=\{5\}$. Thus, we have $\mathbf{F}_{2}^{(1)}(.,\{1\})=[1,5]^\Tsf$; for the second column of $\mathbf{F}_{2}^{(1)}$, since $\Dc_2$ is not available at user $(1,1)$, we require $\mathbf{S}^{(1,1)}_{2,1}\mathbf{F}_{2}^{(1)}(.,\{2\})=0$. Thus, we can solve $\mathbf{F}_{2}^{(1)}(.,\{2\})=[1,4]^\Tsf$. The remaining columns of $\mathbf{F}_{2}^{(1)}$ can be obtained in the same manner. By choosing
$
\mathbf{B}_{2,1}=\begin{bmatrix}
3 & 2 & 3
\end{bmatrix},
%\label{key coefficient matrix B1}
$
whose elements are uniformly i.i.d. over $\mathbb{F}_\qsf$, we have 
\begin{align}
\negmedspace\negmedspace \left[
\begin{NiceArray}{c|[tikz={dash pattern=on 4pt off 3pt}]c}
\mathbf{F}_{2}^{(1)}   & \mathbf{B}_{2}^{(1)} 
\end{NiceArray}
\right] \negmedspace  \mathbf{W}^\prime \negmedspace\negmedspace = \negmedspace\negmedspace
 \left[
    \begin{NiceArray}{cccccc|[tikz={dash pattern=on 4pt off 3pt}]ccc}
       1 & 1 & 6 & 4 & 0 & 1 & 1 & 1 & 2\\
  5&4 &6& 4& 0& 3 &3 &2& 3
    \end{NiceArray}
    \right] \negmedspace
  \mathbf{W}^\prime.
\end{align}
So the transmissions of the two users in cluster \(1\) are
%\begin{subequations}
\begin{alignat*}{2}
&X_{1,1} 
=  -W_1+18W_3+12W_4+W_6+N_1+2N_2+5N_3,\\
&X_{1,2}
= -4W_1-3W_2-2W_6-2N_1-N_2-N_3,
\end{alignat*}
%\label{exp:user within 1 transmission}
%\end{subequations}
where the individual keys of user $(1,1)$ and user $(1,2)$ are chosen as
$
 Z_{1,1}=N_1+2N_2+5N_3$  and  $Z_{1,2}=-2N_1-N_2-N_3$, respectively. 
Hence,  relay \(1\) can recover the assigned relay-layer computation  from \((X_{1,1},X_{1,2})\) since \(\mathbf{S}_{2,1}\) is invertible.

For cluster $2$, the assigned relay-layer computation is
\begin{align*}
    \mathbf{S}^{(2)}_1 \left[
    \begin{NiceArray}{c|[tikz={dash pattern=on 4pt off 3pt}]c}
        \mathbf{F}_{1}  & \mathbf{B}_{1}
    \end{NiceArray}
    \right]\mathbf{W}^{\prime}
=\begin{bmatrix}
 0 & 0 & 5 & 3 & -1 & 0 & 1 & 1 & 2
   \end{bmatrix}\mathbf{W}^{\prime}.
\end{align*} 
We let relay \(2\) recover two key-encrypted linear combinations from any two surviving users, i.e., 
\begin{align}
\left[
\begin{NiceArray}{c|[tikz={dash pattern=on 4pt off 3pt}]c}
\mathbf{F}_{2}^{(2)}   & \mathbf{B}_{2}^{(2)} 
\end{NiceArray}
\right]\mathbf{W}^\prime=\left[
\begin{array}{cc}
 \mathbf{S}^{(2)}_1\mathbf{F}_1 &  \mathbf{S}^{(2)}_1 \mathbf{B}_1\\
\mathbf{V}_{2,2} & \mathbf{B}_{2,2}
\end{array}
\right]
  \mathbf{W}^\prime.
\end{align}
Similarly, we generate 
  $   \mathbf{S}_{2,2}\triangleq
    \begin{bmatrix}
        \mathbf{S}^{(2,1)}_{2,2}\\
        \mathbf{S}^{(2,2)}_{2,2}\\
        \mathbf{S}^{(2,3)}_{2,2}
    \end{bmatrix}= \begin{bmatrix}
       2 & -1 \\
4 & 3 \\
-1 & 1 
    \end{bmatrix}$,
whose elements are uniformly i.i.d. over $\mathbb{F}_\qsf$.
Each user $(2,v)$, where $v\in[3]$, transmits
$    X_{2,v}
=
\mathbf{S}^{(2,v)}_{2,2}
\left[
\begin{NiceArray}{c|[tikz={dash pattern=on 4pt off 3pt}]c}
\mathbf{F}_{2}^{(2)}   & \mathbf{B}_{2}^{(2)} 
\end{NiceArray}
\right]\mathbf{W}^\prime$,
where the coefficients corresponding to the datasets that are not in $\Dc^{(2,v)}$ must be `0'. Thus,  \(\mathbf{V}_{2,2}\) is chosen in the above manner so that the encodability constraints of the three users are satisfied. By choosing $ \mathbf{B}_{2,2}= \begin{bmatrix}
  4 & 5 & 3
    \end{bmatrix}$ with  elements  uniformly i.i.d. over $\mathbb{F}_\qsf$, we have
\begin{align*}
\left[
\begin{NiceArray}{c|[tikz={dash pattern=on 4pt off 3pt}]c}
\mathbf{F}_{2}^{(2)}   & \mathbf{B}_{2}^{(2)} 
\end{NiceArray}
\right]\mathbf{W}^\prime=
 \left[
    \begin{NiceArray}{cccccc|[tikz={dash pattern=on 4pt off 3pt}]ccc}
      0 & 0 & 5 & 3 & -1 & 0 & 1 & 1 & 2\\
  0&0 &2& 6& -1& 0 &4& 5&3
    \end{NiceArray}
    \right]
  \mathbf{W}^\prime.
\end{align*}
Then the transmissions of the three users in cluster \(2\) are
\begin{subequations}
\begin{alignat}{2}
&X_{2,1}
= 8W_3-W_5-2N_1-3N_2+N_3,\\
&X_{2,2} \negmedspace
=  \negmedspace 26W_3 \negmedspace + \negmedspace 30W_4- \negmedspace 7W_5+ \negmedspace 16N_1+ \negmedspace 19N_2+\negmedspace 17N_3,\\
&X_{2,3}
= -3W_3+3W_4+3N_1+4N_2+N_3,
\end{alignat}
\label{exp:user within 2 transmission}
\end{subequations}
where the individual keys of user $(2,1)$, user $(2,2)$, and user $(2,3)$ are chosen as
$
    Z_{2,1}=-2N_1-3N_2+N_3$, $Z_{2,2}=16N_1+19N_2+17N_3$, and $Z_{2,3}=3N_1+4N_2+N_3$, respectively.

Next we explain intuitively the security of the proposed scheme, while the  information theoretic proof is given in Appendices~\ref{sec:proof of relay security} and~\ref{sec:proof of server security}.

{\bf Security against relays.} 
For the security against relays, recall that $\Tsf_1=\Tsf_2=1$, i.e., at most one user in each cluster may collude. By construction, each relay receives two key-encrypted linear combinations from the surviving users. 
Since both $\mathbf B_{2}^{(1)}$ and $\mathbf B_{2}^{(2)}$ are row-wise full rank, the local observations at each relay are perfectly protected by independent keys. 
Therefore, without collusion, each relay cannot infer any information
about the users’ inputs.

We next consider the case where a relay colludes with users. 
For example, consider relay $2$ colluding with users $(1,1)$ and $(2,1)$, i.e.,
$\Tc=\{(1,1),(2,1)\}$.
The key combinations contained in the observations of relay $2$ are
\begin{align}
\mathbf \Ksf_{2,1}=N_1+N_2+2N_3,\
\mathbf \Ksf_{2,2}=4N_1+5N_2+3N_3 .
\end{align}
The colluding user $(2,1)$ reveals the key
$Z_{2,1}=-2N_1-3N_2+N_3$,
which lies in the span of $\{ \Ksf_{2,1},\Ksf_{2,2}\}$. Hence, relay $2$ can remove the corresponding masking key and recover
$8W_3-W_5$ from its observations. 
However, $W_3$ and $W_5$ are already available from the colluding users' side information.  
Therefore, no information about the non-colluding datasets is revealed.
Meanwhile, the key information revealed by user $(1,1)$ is linearly independent of
$\{\Ksf_{2,1},\Ksf_{2,2}\}$.
Consequently, this additional key information cannot help relay $2$ eliminate the masking keys protecting the messages of non-colluding users.
Thus, the security of relay $2$ is guaranteed.

The same argument applies to relay $1$ and all possible choices of the colluding users. 
Therefore, the security against relays is satisfied for all relays.

{\bf Security against server.} 
In~\eqref{eq:eg-sub-tasks}, besides the desired aggregation, the server also obtains one  key-encrypted linear combination from the relay
transmissions. In addition, the server may obtain the key information from the colluding users. For example, consider
the case where the server colludes with users \((1,1)\) and \((2,1)\). 
% By denoting $[T_1;T_{1,1}]= \left[
%     \begin{NiceArray}{c|[tikz={dash pattern=on 4pt off 3pt}]c}
%         \mathbf{F}_{1}  & \mathbf{B}_{1}
%     \end{NiceArray}
%     \right]\mathbf{W}^\prime,$ where $T_1=\sum_{k\in[6]}W_k$, we have
% \begin{subequations}
% \begin{align}
% &I\left(\{Y_u\}_{u\in[2]};\{W_k\}_{k\in[6]}
% \middle|
% \sum_{k\in[6]}W_k,W_{\Tc},Z_{\Tc}\right)\\
% &\leq I\left(\{T_1,T_{1,1}\};\{W_k\}_{k\in[6]}
% \middle|
% \sum_{k\in[6]}W_k,W_{\Tc},Z_{\Tc}\right)\label{F,F_{1,1}}\\
% &=I\left(T_{1,1};\{W_k\}_{k\in[6]}
% \middle|
% \sum_{k\in[6]}W_k,W_{\Tc},Z_{\Tc}\right)\\
% &= H\left(T_{1,1}\middle| \sum_{k\in[6]}W_k,W_{\Tc},Z_{\Tc}\right)\nonumber\\&-
% H\left(T_{1,1}\middle| \{W_k\}_{k\in[6]},\sum_{k\in[6]}W_k,W_{\Tc},Z_{\Tc}\right)\\
% &\le \Lsf-
% H\left(N_1+N_2+2N_3\mid Z_{1,1},Z_{2,1}\right)\\
% &=\Lsf-\Lsf\\
% &=0,\label{independent of the keys}
% \end{align}
% \end{subequations}
% where~\eqref{F,F_{1,1}} follows since $\{Y_u\}_{u\in[2]}$ are in the linear space spanned by  $\{T_1,T_{1,1}\}$ and thus are determined by $\{T_1,T_{1,1}\}$;~\eqref{independent of the keys} follows since 
The key \(N_1+N_2+2N_3\) used to mask the linear combination is linearly
independent of the colluding users' keys $Z_{1,1}$ and $Z_{2,1}$.
So the server obtains no information
about the non-colluding datasets beyond the desired sum.
The same argument applies to all possible cases of  collusion. Hence, the security against server is guaranteed.

Hence, the proposed scheme achieves
\((\Rsf_1,\Rsf_2,\Rsf_Z)=(1,1,1)\),  satisfying the security against server and relays.

% In this section, we propose the achievable hierarchical secure aggregation scheme with $\Tsf^{(u)}\leq \Vsf_u-\rsf_2^{(u)}$ for any relay $u\in[\Usf]$. 
% In the following, we present an example to illustrate the main idea.
% The hierarchical secure aggregation scheme consists of an aggregation server, $\Usf$ relays, and multiple users. For each $u \in [\Usf]$, relay $u$ is associated with $\Vsf_u$ users, among which there are at most $\ssf_u < \rsf_2^{(u)}$ stragglers. The relays are reliable, i.e., $\ssf_1 = 0$. The server and each relay may collude with at most $\Tsf_u \leq \Vsf_u $ users in cluster $u$.

% Given an arbitrary data assignment, we have the minimum repetition factor $\rsf_2^{(u)}$ in cluster $u$ and $\rsf_1$ among all relays, 
% based on which $\msf_2^{(u)}$ and $\msf_1$ are determined. 
\subsection{General Scheme}
\label{general sheme design}
Given an arbitrary data assignment, the minimum replication factors in cluster $u$ and across all relays are denoted by $\rsf_2^{(u)}$ and $\rsf_1$, based on which $\nsf_2=\min_{u\in[\Usf]}\msf_2^{(u)}=\min_{u\in[\Usf]}(\rsf_2^{(u)}-\ssf_2^{(u)})$ and $\msf_1=\rsf_1$ are determined. For each $k\in[\Ksf]$,
we divide  $W_k$ into $\msf =\msf_1\nsf_2$ non-overlapping and equal-length pieces, each containing $\Lsf^\prime=\frac{\Lsf}{\msf}$ i.i.d. uniform symbols over $\mathbb{F}_\qsf$, denoted by $W_k = (W_{k,1},\ldots,W_{k,\msf})$.  
We introduce $\left(\sum_{i\in[\Usf]}\Usf_i-\msf_1\nsf_2\right)$ uniformly i.i.d. keys (each also with $\Lsf^\prime$ symbols) over $\mathbb{F}_\qsf^{\Lsf^\prime}$ as $Z_\Sigma$, and hence $\Rsf_{Z_\Sigma}=\frac{\sum_{i\in[\Usf]}\Usf_i}{\msf_1\nsf_2}-1$.\footnote{
Following the key construction in~\cite{zhao2023secure}, we generalize
the ``total number of users minus one'' principle by choosing the source-key dimension as the total number of  surviving
user-layer transmissions minus the number of the desired
aggregation dimensions.}
% \footnote{The source key rate $\Rsf_{Z_\Sigma}$ characterizes how many symbols $Z_\Sigma$ contains per input symbol. Due to the limitation of pages, related result and discussion could be found in~\cite{hierarchical2026collusion}.}
 We denote $\mathbf{W}\triangleq [W_{1,1},\cdots,W_{\Ksf,1},W_{1,2},\cdots,W_{\Ksf,\msf}]^{\Tsf}$, $\mathbf{N}\triangleq [N_1,\cdots,N_{\msf\Rsf_{Z_\Sigma}}]^{\Tsf}$. As explained in Footnote~\ref{foot:q large}, by taking $\Lsf$ large enough, $\qsf$ can be assumed to be large enough.

{\bf Relay-to-Server layer.} 
To achieve $\Rsf_1 \geq \frac{1}{\msf_1}$ in Theorem~\ref{thm: achievability}, we let each relay transmit  $\sf \nsf_2$ linear combinations of the data pieces and the keys to the server, where each linear combination contains $\Lsf^\prime$ symbols. Hence, after receiving the messages from all $\Usf$ relays, the server obtains a total of $\Usf\nsf_2$ linear combinations, which can be written as
$
 \left[
    \begin{NiceArray}{c|[tikz={dash pattern=on 4pt off 3pt}]c}
        \mathbf{F}_1 & \mathbf{B}_1
    \end{NiceArray}
    \right] \mathbf{W}^\prime =\left[
\begin{array}{cc}
\mathbf{F} & \mathbf{0} \\
\mathbf{V} & \mathbf{B}
\end{array}
\right]
  \begin{bmatrix}
        \mathbf{W}\\
        \mathbf{N}
    \end{bmatrix},
$
where the received coefficient matrix is divided into two parts. The upper part consists of the first $\msf$ rows, where $\mathbf{F}$ is determined by the desired gradient aggregation task. The lower part consists of the remaining rows, which are introduced as virtual key-encrypted linear combinations.
In this lower part, the top-layer key coefficient matrix $\mathbf{B}$ can be generated directly as a full-rank matrix with elements uniformly i.i.d. over $\mathbb{F}_\qsf$ to ensure the security against server. Therefore, it remains to design the virtual demand matrix $\mathbf{V}$  such that the encodability constraints of all relays are satisfied.

% Let
% \begin{align}
%       \left[
%     \begin{NiceArray}{c|[tikz={dash pattern=on 4pt off 3pt}]c}
%         \mathbf{F}_1 & \mathbf{B}_1
%     \end{NiceArray}
%     \right]
%     \triangleq
%     \left[
%     \begin{array}{cc}
%         \mathbf{F} & \mathbf{0}\\
%         \mathbf{V} & \mathbf{B}
%     \end{array}
%     \right],
%      \mathbf{W}^\prime \triangleq 
%     \begin{bmatrix}
%         \mathbf{W}\\
%         \mathbf{N}
%     \end{bmatrix}.
% \end{align}
We let each relay $u\in[\Usf]$ transmit
\begin{align}
    Y_u
    =
    \mathbf{S}^{(u)}_1
    \left[
    \begin{NiceArray}{c|[tikz={dash pattern=on 4pt off 3pt}]c}
        \mathbf{F}_1 & \mathbf{B}_1
    \end{NiceArray}
    \right]
    \mathbf{W}^\prime,
    \label{eq: relay transmission}
\end{align}
where $\mathbf{S}^{(u)}_1$ is of dimension
$\nsf_2\times \Usf\nsf_2$.
We choose the  matrix
$
    \mathbf{S}_1
    \triangleq
    \begin{bmatrix}
        \mathbf{S}^{(1)}_1\\
        \vdots\\
        \mathbf{S}^{(\Usf)}_1
    \end{bmatrix}
$
as a full-rank matrix of dimension
$\Usf\nsf_2\times \Usf\nsf_2$, whose elements are uniformly i.i.d. over $\mathbb{F}_{\qsf}$.
The encodability constraint requires that relay $u$ only forms linear combinations of the datasets available in cluster $u$. Therefore, the coefficients corresponding to the datasets that are not in $\Dc^{(u)}$ must be `0'. For each dataset $\Dc_k\in\Dc\setminus\Dc^{(u)}$, we require
\begin{align}
    \mathbf{S}^{(u)}_1
    \mathbf{F}_1\bigl(.,\{k+(j-1)\Ksf\}\bigr)
    =
    \boldsymbol{0}, \forall j\in[\msf].
    \label{eq:relay-encodability}
\end{align}
For each dataset $\Dc_k\in\Dc\setminus\Dc^{(u)}$, it is unavailable at most $\Usf-\rsf_1$ relays, yielding at most $\nsf_2(\Usf-\rsf_1)$ zero-forcing constraints. The number of rows of the virtual demand matrix is $\Usf\nsf_2-\msf=\nsf_2(\Usf-\rsf_1)$, where $\msf=\msf_1\nsf_2=\rsf_1\nsf_2$. Hence, the matrix $\mathbf{V}$ has enough degrees of freedom to satisfy all encodability constraints. 
After receiving the messages from all relays, the server can recover
$
    \left[
    \begin{NiceArray}{c|[tikz={dash pattern=on 4pt off 3pt}]c}
        \mathbf{F}_1 & \mathbf{B}_1
    \end{NiceArray}
    \right]
    \mathbf{W}^\prime
    =
    \mathbf{S}_1^{-1}
    \begin{bmatrix}
        Y_1\\
        \vdots\\
        Y_{\Usf}
    \end{bmatrix}.
   % \label{eq:server-recover-F1B1}
$
Therefore, the server can recover the aggregation result from these first $\msf$ linear combinations.

{\bf User-to-Relay layer.}
Then, for each relay $u\in[\Usf]$, we assign the key-encrypted linear combinations in~\eqref{eq: relay transmission} as the computational task to be recovered. By the converse bound $\Rsf_2\geq \frac{1}{\msf_1\nsf_2}$, each user transmits one linearly independent message, each of length $\Lsf^\prime$. From the surviving users, relay $u$ receives $\Usf_u$ linear combinations, which can be written as
\begin{align}
\left[
\begin{NiceArray}{c|[tikz={dash pattern=on 4pt off 3pt}]c}
\mathbf{F}_{2}^{(u)}  & \mathbf{B}_{2}^{(u)}
\end{NiceArray}
\right]\mathbf{W}^\prime=\left[
\begin{array}{cc}
 \mathbf{S}^{(u)}_1\mathbf{F}_1 &  \mathbf{S}^{(u)}_1 \mathbf{B}_1\\
\mathbf{V}_{2,u} & \mathbf{B}_{2,u}
\end{array}
\right]
  \begin{bmatrix}
        \mathbf{W}\\
        \mathbf{N}
    \end{bmatrix}.
\end{align}
The first $\nsf_2$  rows are predetermined and correspond exactly to the assigned task in~\eqref{eq: relay transmission}, whereas the remaining rows are introduced as virtual key-encrypted linear combinations. To construct these virtual rows, the lower-layer key coefficient matrix $\mathbf{B}_{2,u}$ is generated as a full-rank matrix with elements uniformly i.i.d. over $\mathbb{F}_\qsf$, and the virtual demand matrix $\mathbf{V}_{2,u}$ is designed accordingly to ensure that the encodability constraints of all users within cluster $u$ are satisfied.

Each user $(u,v)$ transmits
\begin{align}
    X_{u,v}
   =\mathbf{S}^{(u,v)}_{2,u}
    \left[\begin{NiceArray}{c|[tikz={dash pattern=on 4pt off 3pt}]c}
        \mathbf{F}_{2}^{(u)}  & \mathbf{B}_{2}^{(u)} 
        \end{NiceArray}\right]\mathbf{W^\prime},
    \label{eq: user transmission}
\end{align}
where $\mathbf{S}^{(u,v)}_{2,u}$ is of dimension $1\times \Usf_u$.
For each cluster $u$,  we choose the  matrix
$
    \mathbf{S}_{2,u}
    \triangleq
    \begin{bmatrix}
      \mathbf{S}^{(u,1)}_{2,u}\\
        \vdots\\
       \mathbf{S}^{(u,\Vsf_u)}_{2,u}
    \end{bmatrix},
$
whose elements are uniformly i.i.d. over $\mathbb{F}_{\qsf}$. The encodability constraint requires that user $(u,v)$ only forms linear combinations of the datasets assigned to it. Therefore, the coefficients corresponding to the datasets that are not in $\Dc^{(u,v)}$ must be `0'. For each dataset $\Dc_k\in\Dc^{(u)}\setminus\Dc^{(u,v)}$, we require
\begin{align}
    \mathbf{S}^{(u,v)}_{2,u}
    \mathbf{F}_{2}^{(u)}\bigl(.,\{k+(j-1)\Ksf\}\bigr)
    =
    \boldsymbol{0}, \forall j\in[\msf].
    \label{eq:user-encodability}
\end{align}
For each dataset $\Dc_k\in\Dc^{(u)}\setminus\Dc^{(u,v)}$, it is unavailable at most $\Vsf_u-\rsf_2^{(u)}$ users, yielding at most $(\Vsf_u-\rsf_2^{(u)})$ zero-forcing constraints. The number of rows of the virtual demand matrix received by relay $u$ is 
$
\Usf_u-\nsf_2\ge \Usf_u-\msf_2^{(u)}\nonumber=\Vsf_u-\rsf_2^{(u)},
$
where $\msf_2^{(u)}\ge\nsf_2$.
Hence, the virtual demand matrix  has enough degrees of freedom to satisfy all encodability constraints.
%\hfill $\square$ 

We next verify the security.
For each cluster $i\in[\Usf]$, let $\mathbf S_{2,i}^{\Tc_i}$ denote the submatrix of $\mathbf S_{2,i}$ formed by selecting the rows corresponding to the colluding users in $\Tc_i$, where $|\Tc_i|\leq\Tsf_i$. More specifically, 
$
    \mathbf S_{2,i}^{\Tc_i}
    =
    \begin{bmatrix}
        \mathbf S_{2,i}^{\Tc_i(1)} \\
        \vdots \\
        \mathbf S_{2,i}^{\Tc_i(|\Tc_i|)}
    \end{bmatrix},
$
where $\Tc_i(j)$ denotes the \(j\)-th colluding user in cluster \(i\).
For each cluster $i\in[\Usf]$, define the key coefficient matrix revealed by
the colluding users in cluster $i$ as
$
    \mathbf K_{\Tc}^{(i)}
    \triangleq
    \begin{cases}
    \mathbf S_{2,i}^{\Tc_i}\mathbf B_2^{(i)},
    & |\Tc_i|<\Usf_i,\\[1mm]
    \mathbf B_2^{(i)},
    & |\Tc_i|\geq \Usf_i.
    \end{cases}
$

{\bf Security against relays.}
For relay $u\in[\Usf]$, if cluster $u$ is a colluding cluster, i.e., $|\Tc_u|>\Vsf_u-\rsf_2^{(u)}$, then all inputs associated with $\Dc^{(u)}$ are already contained in the side information $W_{\Tc}$. Thus, the security against relays holds trivially for this cluster. It remains to consider the case $|\Tc_u|\le \Vsf_u-\rsf_2^{(u)}$.
% For relay $u$,
% in addition to the transmissions from its own users, it may also obtain additional key information through collusion. 
% For cluster $i\in[\Usf]$, if the colluding users are not sufficient to recover the whole local key space, i.e.,
% $|\Tc_i|<\Usf_i,$
% then relay $u$ can only obtain the key combinations $\mathbf S_{2,i}^{\Tc_i}\mathbf B_2^{(i)}\mathbf N$ transmitted by the colluding users.
% On the other hand, if
% $|\Tc_i|\geq \Usf_i,$ the colluding users provide enough independent transmissions to recover the whole local key space of cluster $i$, relay $u$ can obtain
% $\mathbf B_2^{(i)}\mathbf N.$ Therefore, we define
% $
% \Fc_u
% \triangleq
% \{i\in[\Usf]\setminus\{u\}:|\Tc_i|\geq \Usf_i\}
% $
% as the set of external clusters whose whole local key spaces can be recovered through collusion, and define
% $
% \Pc_u
% \triangleq
% \{i\in[\Usf]\setminus\{u\}:|\Tc_i|<\Usf_i\}
% $
% as the set of external clusters from which relay $u$ can only obtain the key combinations directly transmitted by the colluding users.
The total key coefficient matrix revealed by the external colluding users is
$
    \mathbf K_{u,\Tc}
    \triangleq[ \mathbf K_{\Tc}^{(1)};\ldots; \mathbf K_{\Tc}^{(u-1)};\mathbf K_{\Tc}^{(u+1)};\ldots; \mathbf K_{\Tc}^{(\Usf)}].
    % \begin{bmatrix}
    %     \mathbf K_{\Tc}^{(1)}\\
    %     \vdots\\
    %      \mathbf K_{\Tc}^{(u-1)}\\
    %       \mathbf K_{\Tc}^{(u+1)}\\
    %        \vdots\\
    %     \mathbf K_{\Tc}^{(\Usf)}
    % \end{bmatrix}.
$
It is proved in Appendix~\ref{sec:proof of the information theoretic security for relay} that, if Constraint~\ref{con: relay security} is satisfied,  the information theoretic security against relays is guaranteed. %The  proof is given in Appendix~\ref{sec:proof of the information theoretic security for relay}.
\begin{constraint}[Security against relays with collusion]\label{con: relay security}
For each relay \(u\in[\Usf]\) and each collection of colluding user sets \(\{\Tc_i:i\in[\Usf]\setminus\{u\}\}\) satisfying  $|\Tc_i|\leq\Tsf_i$,  the matrix
$ 
\begin{bmatrix}
\mathbf B_2^{(u)}\\
 \mathbf K_{u,\Tc}
\end{bmatrix}
%\label{matrix in constrain 1}
$ 
has row-wise full rank.
\end{constraint}
In Appendix~\ref{appendix: constraint 1 proof}, we   further prove that Constraint~\ref{con: relay security} holds for the proposed scheme. Hence, the information theoretic security against relays is guaranteed.

{\bf Security against server.}
Recall that  the set of colluding clusters is
$\Bc
    \triangleq
    \left\{
    u\in[\Usf]:
    |\Tc_u|>\Vsf_u-\rsf_2^{(u)}
    \right\}$. For each $u\in\Bc$, every dataset assigned to cluster $u$ is stored by at least one colluding user, and thus the corresponding inputs are already contained in the side information $W_{\Tc}$. We further define
$
    \mathbf S_{1}^{\Bc}
    \triangleq
    \begin{bmatrix}
        \mathbf S_1^{\Bc(1)}\\
        \vdots\\
        \mathbf S_1^{\Bc(|\Bc|)}
    \end{bmatrix},
$
whose dimension is $|\Bc|\nsf_2\times\Usf\nsf_2$.
Let $\mathbf K_{\Tc}$ denote the total key coefficient matrix revealed by the colluding users, i.e.,
$
    \mathbf K_{\Tc}
    \triangleq[ \mathbf K_{\Tc}^{(1)};\ldots; \mathbf K_{\Tc}^{(\Usf)}].
    % \begin{bmatrix}
    %     \mathbf K_{\Tc}^{(1)}\\
    %     \vdots\\
    %     \mathbf K_{\Tc}^{(\Usf)}
    % \end{bmatrix}.
$ 
In Appendix~\ref{sec:proof of the information theoretic security for server} we prove that, 
if Constraint~\ref{con: server security} is satisfied,  the information theoretic security against server is guaranteed. %The  proof is given in Appendix~\ref{sec:proof of the information theoretic security for server}.

\begin{constraint}[Security against server with collusion]\label{con: server security} For each collection of colluding user sets \(\{\Tc_i:i\in[\Usf]\}\) satisfying  $|\Tc_i|\leq\Tsf_i$, we have 
$
    \operatorname{rowspan}(\mathbf B)
    \cap
    \operatorname{rowspan}(\mathbf K_{\Tc})
    \subseteq
    \operatorname{rowspan}
    \left(
        \mathbf S_{1}^{\Bc}\mathbf B_1
    \right).
$ 
\end{constraint}
In Appendix~\ref{proof of constraint 2}, we  prove that Constraint~\ref{con: server security} holds for the proposed scheme. Hence, the information theoretic security against server is guaranteed.

% The proofs of Constraint~\ref{con: relay security} and~\ref{con: server security} are provided in 
% Appendix~\ref{appendix: constraint 1 proof}-~\ref{appendix: constraint 2 proof}.
% the extended version of this paper~\cite{hierarchical2026collusion}.

{\bf Performance.} 
Each relay transmits $Y_u$ in~\eqref{eq: relay transmission}, consisting of $\sf n_2$ messages of length $\Lsf^\prime$; thus $\Rsf_1=\frac{1}{\msf_1}$. Each user $(u,v)$ sends one message of length $\Lsf^\prime$; thus  $\Rsf_2^{(u)}=\frac{1}{\msf_1\nsf_2}$.
%Both Constraint~\ref{con: relay security} and Constraint~\ref{con: server security} are satisfied.

\begin{rem}
\label{rem:ClargerK} 
For the second regime where   $\msf_1\geq2$ and  $\sum_{u\in[\Usf]}
\mathbf{1}\{
\Tsf_u>\Vsf_u-\rsf_2^{(u)}\}
\ge\Usf-\rsf_1$,  the proposed scheme follows a construction similar to that in the first regime. We divide each $W_k$  into $\msf=(\msf_1-1)\nsf_2$ non-overlapping
and equal-length pieces, each containing $\Lsf^\prime=\frac{\Lsf}{(\msf_1-1)\nsf_2}$ symbols. Each user transmits one linear combination of length $\Lsf^\prime$ to its associated relay, while each relay transmits $\nsf_2$ linear combinations to the server. Hence,
$\Rsf_1=\frac{1}{\msf_1-1},
\Rsf_2=\Rsf_Z=\frac{1}{(\msf_1-1)\nsf_2}$.
The encodability follows from the same virtual-demand construction. At the user-to-relay layer, the number of virtual-demand dimensions satisfies
$\Vsf_u-\ssf_2^{(u)}-\nsf_2
\geq
\Vsf_u-\rsf_2^{(u)}$,
since $\msf_2^{(u)}=\rsf_2^{(u)}-\ssf_2^{(u)}\geq\nsf_2$. At the relay-to-server layer, the number of virtual-demand dimensions satisfies
$\Usf\nsf_2-(\msf_1-1)\nsf_2
\geq
(\Usf-\rsf_1)\nsf_2$. Hence,
 the available virtual-demand dimensions are sufficient to satisfy all  encodability constraints at both layers.
The security proof follows a similar argument as in the first regime. For any cluster satisfying
$|\Tc_u|>\Vsf_u-\rsf_2^{(u)}$,
all datasets assigned to this cluster are already contained in the colluding side information, and hence the security constraint holds trivially. For the remaining relays, if $
\sum_{u\in[\Usf]}
\mathbf{1}\left\{
|\Tc_u|>\Vsf_u-\rsf_2^{(u)}
\right\}
>\Usf-\rsf_1
$, the colluding datasets available at the considered relay cover all datasets in the system, and hence the security constraint is also satisfied. Otherwise, the additional keys revealed through collusion are linearly independent of the keys protecting the messages received by the relay. Therefore, the security against relays is guaranteed. For the server, the key information revealed through collusion is either linearly independent of the keys protecting the messages received by the server or associated only with datasets already contained in the colluding side information. Therefore, the server cannot obtain any additional information about the non-colluding datasets. Hence, the security against server is also satisfied.
% If
% $\sum_{u\in[\Usf]}
% \mathbf{1}\{\Tsf_u>\Vsf_u-\rsf_2^{(u)}\}>
% \Usf-\rsf_1$, all datasets are already contained in the colluding side information, and the security constraint holds trivially. It remains to consider the boundary case where $\sum_{u\in[\Usf]}
% \mathbf{1}{\Tsf_u>\Vsf_u-\rsf_2^{(u)}}=
% \Usf-\rsf_1$. In this case, for each effective relay, the key information revealed by the $\Usf-\rsf_1$ colluding clusters, together with the key dimensions protecting its own transmission, are linearly independent. Hence, the relay-security argument follows in the same manner as in the first regime. Similarly, for server security, any key information that becomes resolvable through collusion is associated only with datasets already contained in the colluding side information. Therefore, the same security arguments as in the first regime continue to apply.
\end{rem}

% \begin{thm}\label{thm: capacity3}
%     For the hierarchical secure aggregation problem with arbitrary data assignment consisting of a server, $\Usf$ relays with each connected to a disjoint subset of $\Vsf_u$ users, where  there may be $\ssf_u<\rsf_2^{(u)}$ stragglers and $\Tsf^{(u)}>\Vsf_u-\rsf_2^{(u)}$ colluding users in each cluster, the capacity region is  
%     \begin{equation}
%         \mathcal{R}^{\star}=\left\{
%         \left(\Rsf_1,\Rsf_2^{(u)}\right):\Rsf_1\geq \frac{1}{\min\{\msf_1,\Usf-1\}}, \Rsf_2^{(u)}\geq \frac{1}{(\min\{\msf_1,\Usf-1\})(\Usf_u-\Tsf^{(u)})},  \forall u\in[\Usf]\right\}.
%     \end{equation}
  
% \end{thm}

\section{Proof of Theorem~\ref{thm: multi-dimension}: Achievable Scheme}
\label{sec:heterogeneous-extension}
In this section, we extend the heterogeneous linearly separable computation scheme in~\cite{zhang2026distributed} to a hierarchical network with multi-dimensional tasks (i.e., $\Ksf_{\rm c}\geq 1$). Recall that, in our system model, each $W_k$ is partitioned into $p$ non-overlapping and equal-length pieces. Accordingly, we adopt the fractional communication construction in~\cite{zhang2026distributed}, which allows message subpacketization and fractional communication costs. Since the heterogeneous scheme in~\cite{zhang2026distributed} does not consider user dropouts,
we focus on the no-dropout setting in this extension, where all users and all relays remain active. 

The construction is applied in a layer-wise manner: in the relay-to-server layer, each cluster is treated as a virtual user, while in the user-to-relay layer, the users within each cluster are treated as local users. Recall that $\mathbf{W}=[W_{1,1};\ldots;W_{\Ksf,1};
 W_{1,2};\ldots;W_{\Ksf,p}]$. To incorporate secure aggregation, the source keys are viewed as data symbols available to all users. We introduce $\left(p\Rsf_2\sum_{i\in[\Usf]}\Usf_i-p\right)$ uniformly i.i.d. keys (each containing $\Lsf^\prime=\frac{\Lsf}{p}$ symbols) over $\mathbb{F}_\qsf^{\Lsf^\prime}$ as $Z_\Sigma$, i.e., $\mathbf{N}\triangleq [N_1,\cdots,N_{p\Rsf_2\sum_{i\in[\Usf]}\Usf_i-p}]^{\Tsf}$. Let 
$\mathbf{W}^\prime= \begin{bmatrix}
        \mathbf{W}\\
        \mathbf{N}
    \end{bmatrix}.
$ 

Let $\overline{\Cc_{u,k}} \subseteq [\Vsf_u]$ denote the set of indices $v$ in cluster $u$ such that user $(u,v)$ is not assigned dataset $\Dc_k$, and let $\overline{\Cc_{k}} \subseteq [\Usf]$ denote the set of relays that are not assigned dataset $\Dc_k$. From the messages of all relays, the server can  recover $p\Rsf_1\Usf$ linear combinations, i.e.,
$
\left[
\mathbf F_1
\quad
\mathbf B_1
\right]\mathbf{W}^\prime
=
\begin{bmatrix}
\mathbf F & \mathbf 0\\
\mathbf V & \mathbf B
\end{bmatrix}\mathbf{W}^\prime,
%\label{eq:F1_B1_general_heterogeneous}
$
where $\mathbf F\mathbf W$ is the desired task, $\mathbf V$ is the virtual demand matrix to be determined by the encodability constraints of all relays, and $\mathbf B$ is the top-layer key coefficient matrix, whose elements are chosen uniformly i.i.d. over $\mathbb{F}_\qsf$.
%Since the source keys are available to all relays, the key columns do not impose any zero-forcing constraints and hence the elements in $\mathbf B$ are chosen uniformly i.i.d. over $\mathbb{F}_\qsf$.

Following the fractional communication construction in~\cite{zhang2026distributed}, the heterogeneous construction proceeds as follows.
First, we generate  $\{\mathbf S_1^{(u)}:u\in\mathcal G_1^\prime\}$ with elements uniformly i.i.d. over $\mathbb F_\qsf$. Then, for each $k\in[\Ksf]$ and $j\in[p]$, we impose
\begin{align}
\mathbf S_1^{(u)}
\mathbf F_1(.,\{k+(j-1)\Ksf\})
=
\mathbf{0}, \
u\in \overline{\Cc_{k}}\cap\mathcal G_1^\prime.
\end{align}
There are at most $p\Rsf_1|\overline{\Cc_{k}}\cap\mathcal G_1^\prime|\leq p\Rsf_1 t_1$ linearly independent constraints. Recall in~\eqref{eq:server_Kc}, we have $\Ksf_\csf^{\mathrm{srv}}\leq\Rsf_1(\Usf-t_1)$;
the number of rows of the virtual demand matrix is $p\Rsf_1\Usf-p\Ksf_\csf^{\mathrm{srv}}\geq p\Rsf_1 t_1$. These free entries are sufficient to satisfy all the zero-forcing constraints. 

After $\mathbf V$ is fixed, for each $u\notin\mathcal G_1^\prime$, we construct $\mathbf S_1^{(u)}$ from the left null space induced by the unavailable data-piece columns. More precisely, define
\begin{align*}
\overline{\Kc}^{(u)}
\triangleq
\left\{
k+(j-1)\Ksf:
\Dc_k\in\Dc\setminus\Dc^{(u)},j\in[p]
\right\}.
\end{align*}
Then $\mathbf S_1^{(u)}$ is chosen such that
$\mathbf S_1^{(u)}
\mathbf F_1(.,\overline{\Kc}^{(u)})
=
\mathbf 0.$
By Lemma 1 in~\cite{zhang2026distributed}, the matrix
$
\mathbf S_1
=
\begin{bmatrix}
\mathbf S_1^{(1)}\\
\vdots\\
\mathbf S_1^{(\Usf)}
\end{bmatrix}
$
is full rank with high probability, and the server can recover the task $\mathbf{FW}$. 
% Since $\mathbf B$ is full rank, the virtual linear combinations are masked by independent keys. Therefore, beyond the desired task $\mathbf F\mathbf W$, the server obtains no additional information about the inputs, and server security is guaranteed.

For relay $u\in[\Usf]$, the assigned relay-to-server task is
$
Y_u
=
\mathbf S_1^{(u)}
\left[
\mathbf F_1
\quad
\mathbf B_1
\right]
\mathbf W^\prime.
$
From the messages of all users within cluster $u$,  relay $u$ can recover $p\Rsf_{2}\Vsf_u$ linear combinations, 
i.e.,
%which is
$
\left[
\mathbf F_2^{(u)}
\quad
\mathbf B_2^{(u)}
\right]\mathbf W^\prime
=
\begin{bmatrix}
\mathbf S_1^{(u)}\mathbf F_1
&
\mathbf S_1^{(u)}\mathbf B_1\\
\mathbf V_{2,u}
&
\mathbf B_{2,u}
\end{bmatrix}\mathbf W^\prime.
%\label{eq:F2_B2_general_heterogeneous}
$
We can also choose the elements in $\mathbf B_{2,u}$  uniformly i.i.d. over $\mathbb{F}_\qsf$. 
%Similarly, since the source keys are available to all users within cluster $u$, the elements in $\mathbf B_{2,u}$ are chosen uniformly i.i.d. over $\mathbb{F}_\qsf$. 

Then the heterogeneous construction proceeds as follows.
We first generate
$\left\{
\mathbf S_{2,u}^{(u,v)}:
v\in\mathcal G_{2,u}^\prime
\right\}$ with elements
uniformly i.i.d. over $\mathbb F_\qsf$. We then choose $\mathbf V_{2,u}$ so that, for each $k$ satisfying $\Dc_k\in\Dc^{(u)}$ and each  $j\in[p]$, we impose
\begin{align}
\mathbf S_{2,u}^{(u,v)}
\mathbf F_2^{(u)}(.,\{k+(j-1)\Ksf\})
=\mathbf{0}, \
v\in\overline{\Cc_{u,k}}\cap\mathcal G_{2,u}^\prime.
\end{align}
There are at most $p\Rsf_{2}|\overline{\Cc_{u,k}}\cap\mathcal G_{2,u}^\prime|\leq p\Rsf_{2} t_{2,u}$ linearly independent constraints. Recall in~\eqref{eq:local_Kcu} and~\eqref{eq:C1_min_Kcu}, we have $\Rsf_1\leq\Ksf_\csf^{(u)}
\leq
\Rsf_{2}(\Vsf_u-t_{2,u})$; the
number of rows of the virtual demand matrix is $p\Rsf_{2}\Vsf_u-p\Rsf_1\geq p\Rsf_{2}\Vsf_u-p\Ksf_\csf^{(u)}\geq p\Rsf_{2} t_{2,u}$. These free entries are sufficient to satisfy all the zero-forcing constraints. For each $k$ such that  $\Dc_k\notin\Dc^{(u)}$, all the $p$ columns of $\mathbf V_{2,u}$ corresponding to $\Dc_k$
are  set to  `0'.

After $\mathbf V_{2,u}$ is fixed, for each $v\notin\mathcal G_{2,u}^\prime$, we construct $\mathbf S_{2,u}^{(u,v)}$ from the left
null space induced by the unavailable local data-piece columns. More precisely, define
\begin{align*}
\overline{\mathcal K}^{(u,v)}
\triangleq
\left\{
k+(j-1)\Ksf:
\Dc_k\in\Dc^{(u)}\setminus\Dc^{(u,v)},j\in[p]
\right\}.
\end{align*}
Then $\mathbf S_{2,u}^{(u,v)}$ is chosen such that
$
\mathbf S_{2,u}^{(u,v)}
\mathbf F_2^{(u)}(.,\overline{\mathcal K}^{(u,v)})
=
\mathbf 0.
$
Hence, following Lemma 1 in~\cite{zhang2026distributed},
$
\mathbf S_{2,u}
=
\begin{bmatrix}
\mathbf S_{2,u}^{(u,1)}\\
\vdots\\
\mathbf S_{2,u}^{(u,\Vsf_u)}
\end{bmatrix}
$
is full rank with high probability.
Thus, relay $u$ recovers
$
\left[
\mathbf F_2^{(u)}
\quad
\mathbf B_2^{(u)}
\right]\mathbf W^\prime.
$
From the first $p\Rsf_1$ rows, relay $u$ can send the required message $Y_u$ to the server. 

\textit{Performance.}
The above scheme achieves the task dimension
$\Ksf_{\csf}^{\rm srv}=
\min\{\Rsf_1(\Usf-t_1),\Ksf\}$ with communication rates
$(\Rsf_1,\Rsf_2)$, where $\Rsf_1
\triangleq
\min_{u\in[\Usf]}\Ksf_{\csf}^{(u)}=\min_{u\in[\Usf]}\min\left\{
\Rsf_{2}(\Vsf_u-t_{2,u}),\ |\Dc^{(u)}|
\right\}$. 
Therefore, any task dimension satisfying
$\Ksf_{\csf}\leq \Ksf_{\csf}^{\rm srv}$ is achievable.

% Consider the user-to-relay layer in cluster \(u\). Recall that relay \(u\) recovers \(p\Rsf_2\Vsf_u\) linear combinations from the \(\Vsf_u\) users, among which \(p\Rsf_1\) dimensions correspond to the actual relay-layer tasks. The remaining  $p(\Rsf_2\Vsf_u-\Rsf_1)\ge p\Rsf_2 t_{2,u}$
% dimensions correspond to key-encrypted virtual demands, where the key coefficients are chosen uniformly i.i.d. over $\mathbb{F}_\qsf$.
% Now suppose that at most \(\Tsf_u\leq t_{2,u}\) users in cluster \(u\) may collude, it is similar to the case $\Tsf^{(u)}\leq \Vsf_u-\rsf_2^{(u)}$ in Section~\ref{sec: achievability}, the relay security and server security are satisfied since the key information revealed by collusion is protected by the  lower-layer key coefficient matrix. If \(\Tsf_u\ge t_{2,u}\), one needs either to increase the user-to-relay communication cost \(\Csf_2\), so that more virtual key dimensions are created, or to reduce the relay-layer task dimension \(\Csf_1\), so that more of the recovered  dimensions are reserved as virtual demands.
Now suppose that at most \(\Tsf_u\leq t_{2,u}\) users in cluster \(u\) may collude. In this case, the security against server and relays can be guaranteed in the same spirit as the condition \(\Tsf_u\leq \Vsf_u-\rsf_2^{(u)}\) in Section~\ref{sec: achievability}. Specifically, the amount of key information revealed by the colluding users in cluster $u$ is no larger than 
$p\Rsf_2\Tsf_u$, whereas the number of key-encrypted virtual-demand dimensions satisfies
$p(\Rsf_2\Vsf_u-\Rsf_1)
\geq p\Rsf_2 t_{2,u}
\geq
p\Rsf_2\Tsf_u.$
Since the corresponding key coefficients are chosen uniformly i.i.d. over $\mathbb F_\qsf$, these virtual-demand dimensions provide sufficient independent masking dimensions with high probability. Therefore, the security against both server and relays is preserved.

\iffalse 
The situation is different when \(\Tsf_u>t_{2,u}\). In the construction of Section~\ref{sec: achievability}, if \(\Tsf_u>\Vsf_u-\rsf_2^{(u)}\), cluster \(u\) is regarded as a colluding cluster, and the colluding users collectively cover all datasets assigned to this cluster. Therefore, any information revealed through the local key equations only concerns data symbols that are already included in the colluding side information, and no additional information about non-colluding datasets is leaked.
However, in the heterogeneous extension, the parameter
$t_{2,u}$ is defined only with respect to the selected user set \(\mathcal G_{2,u}^\prime\).  It does not imply that, once \(\Tsf_u>t_{2,u}\), all datasets in cluster \(u\) can be safely treated as colluding side information. Consequently, the key equations revealed by more than \(t_{2,u}\) colluding users may contain linear information involving non-colluding datasets; therefore, the security against server and relays is not necessarily guaranteed.

To make the scheme secure against \(\Tsf_u>t_{2,u}\) colluding users, additional local masking dimensions are required. This can be achieved by either increasing the user-to-relay communication cost \(\Rsf_2\), which creates more lower-layer virtual key dimensions, or reducing the relay-layer task dimension \(\Rsf_1\), which reserves more recovered dimensions as virtual demands. In both cases, the number of independent key dimensions available to mask the information revealed by collusion is increased.
\fi 

\section{Conclusions}
\label{sec:conclusions}
This paper studied the information theoretic secure distributed linearly separable computation problem over a three-layer hierarchical network with arbitrary heterogeneous data assignment. For gradient aggregation with \(\Ksf_\csf=1\), we established converse and achievable bounds in the presence of user dropouts and collusion, and proposed a linear-coding-based secure scheme that achieves the optimal two-layer communication rates in one regime and order-optimal communication rates within a factor of \(2\) in the other regime, while guaranteeing the security against server and relays. We further extended the proposed construction to multi-dimensional linearly separable tasks with \(\Ksf_\csf>1\) under the no-dropout setting and established an achievable task dimension determined by the user-level and cluster-level data assignment structures. The resulting scheme satisfies the security requirements against user collusion within the prescribed thresholds. 
Ongoing work   focuses on three directions: characterizing the tradeoff between communication cost and achievable task dimension for multi-dimensional tasks in the presence of user dropouts; determining the minimum total number of keys required to guarantee the security against server and relays; and developing secure schemes that tolerate stronger user collusion.

\appendices
\section{Proof of Theorem~\ref{thm: converse} }\label{appendix:theorem for r_1=U}
\subsection{Proof for the first regime}
\label{proof for the first regime}
We first  prove $\Rsf_1\geq \frac{1}{\msf_1}$.
Choose a dataset $\Dc_i$ such that $|\mathcal N_i|=\msf_1$, and let
$\Ac\triangleq\mathcal N_i$ denote the set of relays assigned $\Dc_i$.
By the decodability constraint, we have
\begin{subequations}
\begin{align}
0&=H\left(\sum_{k\in[\Ksf]} W_k \,\middle|\, \{Y_u\}_{u\in[\Usf]}\right)\\
&\geq H\left(
W_i
\,\middle|\,
\{W_k\}_{k\in[\Ksf]\setminus\{i\}},Z_\Sigma,\{Y_u\}_{u\in[\Usf]}
\right).
\end{align}
\end{subequations}
Hence,
\begin{subequations}
\begin{align}
\Lsf&=H\left(W_i\,\middle|\,
\{W_k\}_{k\in[\Ksf]\setminus\{i\}},Z_\Sigma\right)\\
&=
I\left(
W_i;
\{Y_u\}_{u\in[\Usf]}
\,\middle|\,
\{W_k\}_{k\in[\Ksf]\setminus\{i\}},Z_\Sigma
\right) \\
&=
I\left(
W_i;
\{Y_u\}_{u\in\mathcal A}
\,\middle|\,
\{W_k\}_{k\in[\Ksf]\setminus\{i\}},Z_\Sigma
\right)\\
&\le
H\left(
\{Y_u\}_{u\in\mathcal A}
\,\middle|\,
\{W_k\}_{k\in[\Ksf]\setminus\{i\}},Z_\Sigma
\right)\\
&\le \sum_{u\in\mathcal A}H\left(
Y_u
\,\middle|\,
\{W_k\}_{k\in[\Ksf]\setminus\{i\}},Z_\Sigma
\right)\label{inequation for R2}\\
&\le
\sum_{u\in\mathcal A} H(Y_u). 
\end{align}
\label{YucontainWi}
\end{subequations}
Thus, we  have 
$\max_{u\in\mathcal A}  H(Y_u)\geq \frac{\Lsf}{\msf_1}$ and
$\Rsf_1 \ge \frac{1}{\msf_1}$.

Next, we prove the converse bound on $\Rsf_2\geq\frac{1}{
\msf_1
\nsf_2
}$.
% Now consider a dataset $\Dc_i$ such that
% $|\mathcal N_i|=\msf_1$; we assume the relays in $\Ac\triangleq\mathcal N_i$ are assigned $\Dc_i$. 
% For such a dataset, 
From~\eqref{inequation for R2}, we obtain
\begin{align}
      \sum_{u \in \mathcal{A}} H\left( Y_u \mid \{W_k\}_{k\in[\Ksf]\setminus\{i\}}, Z_\Sigma\right)\geq\Lsf.
\end{align}
Then we let
\begin{align}
    &H\left( Y_{u_1} \mid  \{W_k\}_{k\in[\Ksf]\setminus\{i\}}, Z_\Sigma\right) \nonumber \\&=\max_{u\in\Ac} H\left( Y_u \mid  \{W_k\}_{k\in[\Ksf]\setminus\{i\}}, Z_\Sigma\right) \geq\frac{\Lsf}{\msf_1}.
\end{align}

We assume that in cluster $u_1$,  the dataset $D_i$ has the  minimum replication factor, and the only surviving users $\Hc$, where $\Hc\subseteq\{u_1\}\times[\Vsf_{u_1}], |\Hc|=\msf_2^{(u_1)}$, are assigned $\Dc_i$. 
\begin{subequations}
\begin{align}
&\msf_2^{(u_1)} \Rsf_2 \Lsf\geq H\!\left(
(X_{u_1,v})_{(u_1,v)\in \Hc}
| 
 \{W_k\}_{k\in[\Ksf]\setminus\{i\}}, Z_\Sigma
\right)\\
&\ge
I\!\left(
(X_{u_1,v})_{(u_1,v)\in\Hc}; Y_{u_1}
\,\middle|\,
 \{W_k\}_{k\in[\Ksf]\setminus\{i\}}, Z_\Sigma
\right)\\
&=
H\!\left(
Y_{u_1}
\,\middle|\,
 \{W_k\}_{k\in[\Ksf]\setminus\{i\}}, Z_\Sigma
\right)\nonumber\\&-\underbrace{H\!\left(
Y_{u_1}
\,\middle|\,
 \{W_k\}_{k\in[\Ksf]\setminus\{i\}}, Z_\Sigma,
(X_{u_1,v})_{(u_1,v)\in \Hc}
\right)
}_{\stackrel{}{=} 0}\label{in Bc =0}
\\
&\ge
\frac{\Lsf}{\msf_1},
\end{align} 
where \eqref{in Bc =0} follows since $(X_{u_1,v})_{(u_1,v)\in \Mc_{u_1}\setminus\Hc}$ can be recovered from $\{ \{W_k\}_{k\in[\Ksf]\setminus\{i\}}, Z_\Sigma\}$ and $Y_{u_1}$ is a function of  $(X_{u_1,v})_{(u_1,v)\in \Mc_{u_1}}$.
\label{eq:R2-replication-upper}
\end{subequations}

Thus, we have 
$    \Rsf_2 \ge
\frac{1}{
\msf_1
\msf_2^{(u_1)}
}$.
Therefore, taking the worst-case cluster yields
  $  \Rsf_2 \ge
\frac{1}{
\msf_1
\min_{u\in[\Usf]}
\msf_2^{(u)}
}=
\frac{1}{
\msf_1
\nsf_2
}.$
 %\hfill $\square$ 

Finally, we will prove that  $\Rsf_Z\geq\frac{1}{
\msf_1
\nsf_2
}$.
For any $u\in[\Usf]$ and $v\in[\Vsf_u]$, we have
\begin{subequations}
\begin{align}
&H(Z_{u,v})  \ge H(Z_{u,v}\mid W_{u,v}) \\
&\ge I(X_{u,v};Z_{u,v}\mid W_{u,v}) \\
&= H(X_{u,v}\mid W_{u,v})
   - H(X_{u,v}\mid W_{u,v},Z_{u,v}) \\
&\overset{\eqref{eq: user message generation}}{=} H(X_{u,v}\mid W_{u,v}) \\
&= H(X_{u,v})
   - I(X_{u,v};W_{u,v})\\
&\overset{\eqref{eq: relay security}}{=} H(X_{u,v}), 
\end{align}
\end{subequations}
Thus, we obtain
%  $ \Rsf_2=\frac{\max_{u\in [\Usf],v\in[\Vsf_u]}  H(X_{u,v})}{\Lsf}  \ge \frac{1}{\msf_1\nsf_2}$.
% Therefore,
$\Rsf_{Z} \ge\frac{\max_{u\in [\Usf],v\in[\Vsf_u]}  H(X_{u,v})}{\Lsf} \ge \frac{1}{\msf_1\nsf_2}.$

\subsection{Proof for the second regime}
\label{proof for the second regime}
In the following, we prove the tightened converse bound on $\Rsf_1$ for the case $\sum_{u\in[\Usf]}
\mathbf{1}\{
\Tsf_u\ge\Vsf_u-\ssf_2^{(u)}\}
=\Usf-\rsf_1 $. Define
\begin{align}
\mathcal E
\triangleq
\left\{
u\in[\Usf]:
\Tsf_u\geq\Vsf_u-\ssf_2^{(u)}
\right\}.
\end{align}
Then
 we have $
|\mathcal E|=\Usf-\rsf_1.
$
% Consider any dataset that is not available to the colluding users. Such a dataset cannot be assigned to a cluster satisfying
% $
% \Tsf_u>\Vsf_u-\rsf_2^{(u)},
% $
% because, under this condition, the colluding users in this cluster  collectively have access to every dataset assigned to that cluster. Since exactly \(\Usf-\rsf_1\) clusters satisfy $\Tsf_u>\Vsf_u-\rsf_2^{(u)}$, there are exactly \(\rsf_1\) remaining clusters. Moreover, every dataset is assigned to at least \(\rsf_1\) clusters. Therefore, every dataset unavailable to the colluding users must be assigned to all of these \(\rsf_1\) remaining clusters.
By the security constraint against relays, for each \(u\in[\Usf]\), we have
\begin{subequations}
\begin{align}
&I\!\left(Y_u;\sum_{k\in[\Ksf]} W_k\Big|\{W_{\Tc},Z_{\Tc}\}\right)\\&
\le
 I\Big(\{X_{u,v}\}_{v\in[\Vsf_u]};\{W_k\}_{k\in[\Ksf]}\Big|\{W_{\Tc},Z_{\Tc}\}\Big)
\overset{\eqref{eq: relay security}}{=}0.
\end{align}
\end{subequations}
% Since \(\sum_{k\in[\Ksf]} W_k\) is a deterministic function of \(\{W_k\}_{k\in[\Ksf]}\), it follows that
% \begin{subequations}
% \begin{align}
% &I\!\left(Y_u;\sum_{k\in[\Ksf]} W_k\Big|\{W_{\Tc},Z_{\Tc}\}\right)\\&\le I\!\left(Y_u;\{W_k\}_{k\in[\Ksf]}\Big|\{W_{\Tc},Z_{\Tc}\}\right)=0.
% \end{align}
% \end{subequations}
Hence,
\begin{align}
H\!\left(\negmedspace \sum_{k\in[\Ksf]}\negmedspace W_k \middle| Y_u,W_{\Tc},Z_{\Tc}\right)
\negmedspace = \negmedspace 
H\!\left(\negmedspace \sum_{k\in[\Ksf]}\negmedspace W_k \middle| W_{\Tc},Z_{\Tc}\right).\label{infer1}
\end{align}

On the other hand, the decodability constraint gives
\begin{align}
H\Bigg(\sum_{k\in[\Ksf]}W_k\Big|\left\{Y_{u}\right\}_{u\in[\Usf]}\Bigg)=0.\label{infer2}
\end{align}
Suppose there exists a non-colluding dataset.
For any \(i\in[\Usf]\setminus\Ec\), by~\eqref{infer1} and~\eqref{infer2} we have
\begin{subequations}
\begin{align}
&\Lsf=H\!\left(\sum_{k\in[\Ksf]} W_k\,\middle|\,W_{\Tc},Z_{\Tc}\right)\\
&=
H\!\left(\sum_{k\in[\Ksf]} W_k \,\middle|\, Y_{i},W_{\Tc},Z_{\Tc}\right)\nonumber
\\&-
H\Bigg(\sum_{k\in[\Ksf]}W_k\Big|\left\{Y_{u}\right\}_{u\in[\Usf]},W_{\Tc},Z_{\Tc}\Bigg)\label{reduction of conditional entropy}\\
&=
H\!\left(\sum_{k\in[\Ksf]} W_k \,\middle|\, Y_{i},\left\{Y_{u}\right\}_{u\in\mathcal E},W_{\Tc},Z_{\Tc}\right)\nonumber
\\&-
H\Bigg(\sum_{k\in[\Ksf]}W_k\Big|\left\{Y_{u}\right\}_{u\in[\Usf]},W_{\Tc},Z_{\Tc}\Bigg)\label{message Y from colluding users}\\
&=
I\negmedspace\negmedspace \left(\negmedspace \sum_{k\in[\Ksf]} \negmedspace\negmedspace W_k; \negmedspace \left\{Y_{u}\right\}_{u\in[\Usf]\setminus(\{i\}\cup\mathcal E)} \negmedspace \middle| Y_{i},  \negmedspace \left\{Y_{u}\right\}_{u\in\mathcal E} \negmedspace,W_{\Tc},Z_{\Tc} \negmedspace\negmedspace\right) \\
&\le
H\!\left(\left\{Y_{u}\right\}_{u\in[\Usf]\setminus(\{i\}\cup\mathcal E)} \,\middle|\, Y_{i}, \negmedspace \left\{Y_{u}\right\}_{u\in\mathcal E} \negmedspace,W_{\Tc},Z_{\Tc}\right) \\
&\le
\sum_{u\in[\Usf]\setminus(\{i\}\cup\mathcal E)} H(Y_u),
\end{align}
where~\eqref{reduction of conditional entropy} follows from
\eqref{infer1} and \eqref{infer2}, together with
the fact that conditioning reduces entropy.
Moreover,~\eqref{message Y from colluding users} follows since,
for every $u\in\mathcal E$, all surviving user messages in
cluster $u$ are generated by colluding users and are therefore
determined by $(W_{\Tc},Z_{\Tc})$. Since $Y_u$ is a
deterministic function of these surviving user messages,  conditioning additionally on
$\{Y_u\}_{u\in\mathcal E}$ does not change the conditional
entropy.
\label{m1=U infer}
\end{subequations}

So we  have 
$\max_{u\in[\Usf]\setminus(\{i\}\cup\mathcal E)}  H(Y_u)\geq \frac{\Lsf}{\rsf_1-1}=\frac{\Lsf}{\msf_1-1}$ and thus 
$
\Rsf_1 \ge \frac{1}{\msf_1-1}.   
$

\subsection{Proof for the third regime}
\label{proof for the third regime}
We consider the case $\msf_1=\rsf_1=1$ and $\sum_{u\in[\Usf]}
\mathbf{1}\{
\Tsf_u\ge\Vsf_u-\ssf_2^{(u)}\}
\ge\Usf-\rsf_1$.
Define
\begin{align}
\mathcal{E}
\triangleq
\left\{
u\in[\Usf]:
\Tsf_u\geq \Vsf_u-\ssf_2^{(u)}
\right\}.
\end{align}

For each $u\in\mathcal{E}$, consider the surviving user set $\mathcal{M}_u$ and the colluding user set
$\mathcal{T}_u$ satisfying $\mathcal{M}_u\subseteq\mathcal{T}_u$. This is feasible since
$|\mathcal{M}_u|
=\Vsf_u-\ssf_2^{(u)}
\leq \Tsf_u.$
By the constraint in~\eqref{eq: relay message generation}, 
\begin{align}
H\left(
Y_u
\middle|
W_{\mathcal{T}},Z_{\mathcal{T}}
\right)
=0,
\qquad
u\in\mathcal{E}.
\label{message in E}
\end{align}

Suppose that there exists at least one non-colluding dataset. Then $\mathcal{E}\neq[\Usf]$, and thus $|\mathcal{E}|=\Usf-1$. Let $u^\star=[\Usf]\setminus\mathcal{E}$ denote the unique effective relay. Since every dataset assigned to a cluster in $\mathcal{E}$ is contained in $W_{\mathcal{T}}$, every non-colluding dataset can only be assigned to cluster $u^\star$.

By the security against relay $u^\star$ in~\eqref{eq: relay security}, we have
\begin{align}
I\left(
\{X_{u^\star,v}\}_{v\in[\Vsf_{u^\star}]};
\{W_k\}_{k\in[\Ksf]}
\middle|
W_{\mathcal{T}},Z_{\mathcal{T}}
\right)
=0.
\label{meaasges from users not in E}
\end{align}
Since $Y_{u^\star}$ is a deterministic function of the surviving user messages in cluster $u^\star$ and
$\sum_{k\in[\Ksf]}W_k$ is a function of $\{W_k\}_{k\in[\Ksf]}$, the data-processing inequality gives
\begin{align}
&I\left(
Y_{u^\star};
\sum_{k\in[\Ksf]}W_k
\middle|
W_{\mathcal{T}},Z_{\mathcal{T}}
\right)
=0.
\label{data-processing inequality}
\end{align}
Combining \eqref{message in E} and \eqref{data-processing inequality}, we obtain
\begin{align}
&I\left(
\sum_{k\in[\Ksf]}W_k;
\{Y_u\}_{u\in[\Usf]}
\middle|
W_{\mathcal{T}},Z_{\mathcal{T}}
\right)
\nonumber\\
&=
I\left(
\sum_{k\in[\Ksf]}W_k;
Y_{u^\star},\{Y_u\}_{u\in\mathcal{E}}
\middle|
W_{\mathcal{T}},Z_{\mathcal{T}}
\right)
\nonumber\\
&=
I\left(
\sum_{k\in[\Ksf]}W_k;
Y_{u^\star}
\middle|
W_{\mathcal{T}},Z_{\mathcal{T}}
\right)\nonumber\\&
=0.
\label{inference1}
\end{align}

Since there exists at least one non-colluding dataset and the intermediate outcomes are independent of the keys, 
\begin{align}
H\left(
\sum_{k\in[\Ksf]}W_k
\middle|
W_{\mathcal{T}},Z_{\mathcal{T}}
\right)
=\Lsf.
\label{eq}
\end{align}
Hence, from~\eqref{inference1}
\begin{align}
&H\left(
\sum_{k\in[\Ksf]}W_k
\middle|
\{Y_u\}_{u\in[\Usf]},
W_{\mathcal{T}},Z_{\mathcal{T}}
\right)\nonumber \\&=
H\left(
\sum_{k\in[\Ksf]}W_k
\middle|
W_{\mathcal{T}},Z_{\mathcal{T}}
\right)=\Lsf.
\label{security infer}
\end{align}
On the other hand, by the decodability constraint in~\eqref{eq: decodability},
\begin{align}
H\left(
\sum_{k\in[\Ksf]}W_k
\middle|
\{Y_u\}_{u\in[\Usf]}
\right)
=0,
\end{align}
which further implies
\begin{align}
H\left(
\sum_{k\in[\Ksf]}W_k
\middle|
\{Y_u\}_{u\in[\Usf]},
W_{\mathcal{T}},Z_{\mathcal{T}}
\right)
=0.
\label{decodability infer}
\end{align}
This contradicts \eqref{security infer}. Therefore, the security against relays and the decodability cannot be simultaneously satisfied, and secure distributed linearly separable computation is infeasible when $\msf_1=1$ and $\sum_{u\in[\Usf]}
\mathbf{1}\{
\Tsf_u\ge\Vsf_u-\ssf_2^{(u)}\}
\ge\Usf-\rsf_1$.

\bibliographystyle{IEEEtran}
\bibliography{ref}

\newpage 
\section{Proofs of the information theoretic security against relays}
\label{sec:proof of relay security}
\subsection{Proof of the information theoretic security}
\label{sec:proof of the information theoretic security for relay}
We aim to prove that, if Constraint~\ref{con: relay security} is satisfied, the proposed scheme is information theoretically secure against relays.

For relay \(u\in[\Usf]\), it suffices to consider the case $|\Tc_u|\leq \Vsf_u-\rsf_2^{(u)}<\Usf_u$,  since otherwise all inputs in cluster $u$ are already contained in the side information $W_{\Tc}$. Define

\begin{align}
    \begin{bmatrix}
        F_{u,1}\\
        \vdots\\
        F_{u,\Usf_u}
    \end{bmatrix}
=\left[
\begin{NiceArray}{c|[tikz={dash pattern=on 4pt off 3pt}]c}
\mathbf{F}_{2}^{(u)}  & \mathbf{B}_{2}^{(u)}
\end{NiceArray}
\right]\mathbf{W}^\prime.
\end{align}

To prove the security constraint in~\eqref{eq: relay security}, we have
\begin{subequations}
\begin{align}
& I\Big(\{X_{u,v}\}_{v\in[\Vsf_u]};\{W_k\}_{k\in[\Ksf]}
    \Big|\{W_{\Tc},Z_{\Tc}\}\Big) \\
& \leq
I\Big(\{F_{u,1},\ldots,F_{u,\Usf_u}\};\{W_k\}_{k\in[\Ksf]}
    \Big|\{W_{\Tc},Z_{\Tc}\}\Big) \label{eq:relay-sec-x-to-f}\\
& =
H\Big(\{F_{u,1},\ldots,F_{u,\Usf_u}\}
    \Big|\{W_{\Tc},Z_{\Tc}\}\Big) \nonumber\\
&-
H\Big(\{F_{u,1},\ldots,F_{u,\Usf_u}\}
    \Big|\{W_k\}_{k\in[\Ksf]},\{W_{\Tc},Z_{\Tc}\}\Big) \label{eq:relay-sec-mutual info}\\
& \leq
(\Usf_u-|\Tc_u|)\Lsf^\prime
-
H\Big(\mathbf B_2^{(u)}\mathbf N
    \Big|Z_{\Tc_u},Z_{\Tc_{[\Usf]\setminus\{u\}}}\Big) \label{eq:relay-sec-key-mask}\\
& =
(\Usf_u-|\Tc_u|)\Lsf^\prime
-
(\Usf_u-|\Tc_u|)\Lsf^\prime \label{eq:relay-sec-key-entropy}\\
&=0,
\label{eq:proof-relay-security}
\end{align}
\end{subequations}
where \eqref{eq:relay-sec-x-to-f} follows since 
\(\{X_{u,v}\}_{v\in[\Vsf_u]}\) are in the linear space spanned by
\(\{F_{u,1},\ldots,F_{u,\Usf_u}\}\), and hence are determined by \(\{F_{u,1},\ldots,F_{u,\Usf_u}\}\). \eqref{eq:relay-sec-key-mask} follows since the messages sent by the
colluding users in cluster \(u\) are deterministic functions of \(\{W_{\Tc_u},Z_{\Tc_u}\}\); the side
information \(\{W_{\Tc_u},Z_{\Tc_u}\}\) reveals $\mathbf S_{2,u}^{\Tc_u}[F_{u,1};\ldots; F_{u,\Usf_u}]$. Since  \(\mathbf S_{2,u}\) is generated with elements uniformly i.i.d. over
\(\mathbb F_\qsf\), any \(|\Tc_u|<\Usf_u\) rows of \(\mathbf S_{2,u}\) are
linearly independent with high probability as \(\qsf\to\infty\). Therefore,
conditioning on the colluding users' information removes exactly
\(|\Tc_u|\) independent linear combinations from
\(\{F_{u,1},\ldots,F_{u,\Usf_u}\}\), and at most
\((\Usf_u-|\Tc_u|)\) independent dimensions remain unknown.
Meanwhile, by Constraint~\ref{con: relay security}, the row space of
\(\mathbf B_2^{(u)}\)  is linearly independent of the row space of  $\mathbf K_{u,\Tc}$. Hence, the local key space \(\mathbf B_2^{(u)}\mathbf{N}\) is independent of the
external colluding key
information $Z_{\Tc_{[\Usf]\setminus\{u\}}}\triangleq\mathbf K_{u,\Tc}\mathbf{N}$. Moreover, the colluding users inside cluster \(u\) reveal
\(Z_{\Tc_u}\triangleq\mathbf S_{2,u}^{\Tc_u}\mathbf B_2^{(u)}\mathbf N\), which consists of
\(|\Tc_u|\) independent linear combinations of the local key space.
Therefore, after conditioning on \(Z_{\Tc_u}\) and
\(Z_{\Tc_{[\Usf]\setminus\{u\}}}\), the remaining entropy of
\(\mathbf B_2^{(u)}\mathbf N\) is
\((\Usf_u-|\Tc_u|)\Lsf^\prime\), which gives
\eqref{eq:relay-sec-key-entropy}. 

Hence, the security against relays is proved.

\subsection{Proof of Constraint~\ref{con: relay security}}\label{appendix: constraint 1 proof}
Recall that \(\msf=\msf_1\nsf_2\), where \(\msf_1=\rsf_1\) and \(\nsf_2=\min_{u\in[\Usf]}\msf_2^{(u)}\). In the relay-to-server layer, the server recovers \(\Usf\nsf_2\) linear combinations, among which the first \(\msf\) rows correspond to the desired aggregation task and the remaining \(\Usf\nsf_2-\msf\) rows are virtual key-encrypted combinations with key coefficients uniformly i.i.d. over \(\mathbb F_\qsf\).  Hence, \begin{align} \text{rank}(\mathbf B)=\Usf\nsf_2-\msf . \end{align}
% Hence, the top-layer key coefficient matrix \(\mathbf B_1=[\mathbf{0};\mathbf B]\), where $\mathbf B$ is chosen with elements uniformly i.i.d. over \(\mathbb F_\qsf\), such that 
% \[ \text{rank}(\mathbf B)=\Usf\nsf_2-\msf . \] 

For each cluster \(u\in[\Usf]\), relay $u$ receives $\Usf_u$ linear combinations from the surviving users. The corresponding local key coefficient matrix \(\mathbf B_2^{(u)}\) can be written as 
\begin{align} \mathbf B_2^{(u)} = \begin{bmatrix} \mathbf S_1^{(u)}\mathbf B_1\\ \mathbf B_{2,u} \end{bmatrix}, \end{align}
where the first \(\nsf_2\) rows, $\mathbf S_1^{(u)}\mathbf B_1$, correspond to the relay-layer task assigned to cluster \(u\), and $\mathbf B_{2,u} $ is generated with elements uniformly i.i.d. over \(\mathbb F_\qsf\).

For relay \(u\in[\Usf]\), let \begin{align}\Bc_{u} \triangleq \left\{ i\in[\Usf]\setminus\{u\}: |\Tc_i|>\Vsf_i-\rsf_2^{(i)} \right\}
\end{align} 
denote the set of colluding clusters other than cluster \(u\), and let 
\begin{align} 
\Gc_{u} \triangleq\left\{ i\in[\Usf]\setminus\{u\}: |\Tc_i|\leq\Vsf_i-\rsf_2^{(i)} \right\}
\end{align}
denote the remaining non-colluding clusters. Since the total number of colluding clusters satisfies $|\Bc|\leq \Usf-\rsf_1-1$ and cluster $u$ is non-colluding, we have \[ |\Bc_{u}|\leq \Usf-\rsf_1-1 . \] 

We aim to prove that, for each relay \(u\in[\Usf]\) and each collection of colluding user sets \(\{\Tc_i:i\in[\Usf]\setminus\{u\}\}\) satisfying  $|\Tc_i|\leq\Tsf_i$,
the matrix
\begin{align}
\begin{bmatrix}
\mathbf B_2^{(u)}\\
 \mathbf K_{u,\Tc}
\end{bmatrix}
\label{the proved matrix for relay security}
\end{align}
has row-wise full rank, where $\mathbf K_{u,\Tc}$ denotes the key coefficient matrix revealed to relay $u$ through collusion with users in the external clusters. 

Since
$\Bc_{u}\cup\Gc_{u}
=
[\Usf]\setminus\{u\}$,
the external key information can be decomposed according to the two types of clusters in $\Bc_u$ and $\Gc_u$.

For each cluster $i\in\Gc_u$, we have
\begin{align} |\Tc_i|\leq \Vsf_i-\rsf_2^{(i)}<\Usf_i. 
\end{align} 
For each $i\in\Gc_{u}$, since 
\begin{align} 
\nsf_2\leq \msf_2^{(i)}=\rsf_2^{(i)}-\ssf^{(i)}_2, 
\end{align}
it follows that
\begin{align}
 |\Tc_i|\leq\Vsf_i-\rsf_2^{(i)} \leq \Vsf_i-\ssf^{(i)}_2-\nsf_2=\Usf_i-\nsf_2. 
\end{align}
The key information revealed by the colluding users in cluster $i$ is
$\mathbf S_{2,i}^{\Tc_i}\mathbf B_2^{(i)}\mathbf N$. Using
\begin{align}
    \mathbf B_2^{(i)}
    =
    \begin{bmatrix}
        \mathbf S_1^{(i)}\mathbf B_1\\
        \mathbf B_{2,i}
    \end{bmatrix},
\end{align}
we can write
\begin{align}
    \mathbf S_{2,i}^{\Tc_i}\mathbf B_2^{(i)}
    =
    \mathbf S_{2,i}^{\Tc_i}(.,[\nsf_2])
        \mathbf S_1^{(i)}\mathbf B_1
    +
    \mathbf S_{2,i}^{\Tc_i}(.,[\nsf_2+1:\Usf_i])
        \mathbf B_{2,i}.
    \label{eq:partial-external-key}
\end{align}
The first term in~\eqref{eq:partial-external-key} lies in the top-layer key space, while the second term lies in the space generated
by  $\mathbf B_{2,i}$. Since
$\mathbf S_{2,i}$ is generated with elements uniformly i.i.d. over $\mathbb F_\qsf$ and
$|\Tc_i|\leq \Usf_i-\nsf_2$, the submatrix
$\mathbf S_{2,i}^{\Tc_i}(.,[\nsf_2+1:\Usf_i])$
has rank equal to $|\Tc_i|$ with high probability. Moreover, $\mathbf B_{2,i}$ is generated independently with elements uniformly i.i.d. over $\mathbb F_\qsf$. Therefore, the second term in~\eqref{eq:partial-external-key} contributes $|\Tc_i|$ linearly independent key directions outside the top-layer key space $\mathbf{B_1}$.

Next, consider a cluster $i\in\Bc_u$. Since cluster $i$ is colluding, all datasets assigned to this cluster are already covered by the colluding users. For the purpose of proving  security against relays, we may consider the worst case where the colluding users reveal the whole local key space of cluster $i$, namely
\begin{align}
    \mathbf B_2^{(i)}\mathbf N
    =
    \begin{bmatrix}
        \mathbf S_1^{(i)}\mathbf B_1\\
        \mathbf B_{2,i}
    \end{bmatrix}\mathbf N .
\end{align}

Thus, the matrix in~\eqref{the proved matrix for relay security} can be written as
\begin{align}
\begin{bmatrix}
 \mathbf S_1^{(u)}\mathbf B_1\\
 \mathbf{B}_{2,u}\\
 \mathbf S_1^{(\Bc_{u}(1))}\mathbf B_1\\
 \mathbf B_{2,\Bc_{u}(1)}\\
\vdots\\
 \mathbf S_1^{(\Bc_{u}(|\Bc_{u}|))}\mathbf B_1\\
 \mathbf B_{2,\Bc_{u}(|\Bc_{u}|)}\\
\mathbf S_{2,\Gc_u(1)}^{\Tc_{\Gc_u(1)}}
    \mathbf B_2^{(\Gc_u(1))}\\
\vdots\\
\mathbf S_{2,\Gc_u(|\Gc_u|)}^{\Tc_{\Gc_u(|\Gc_u|)}}
    \mathbf B_2^{(\Gc_u(|\Gc_u|))}
\end{bmatrix}.
\label{the proved matrix for relay security2}
\end{align} 

We check that
\begin{align}
    \begin{bmatrix}
        \mathbf S_1^{(u)}\mathbf B_1\\
        \mathbf S_1^{(\Bc_{u}(1))}\mathbf B_1\\
        \vdots\\
         \mathbf S_1^{(\Bc_{u}(|\Bc_{u}|))}\mathbf B_1
    \end{bmatrix}=\begin{bmatrix}
         \mathbf S_1^{(u)}\\
        \mathbf S_1^{(\Bc_{u}(1))}\\
        \vdots\\
         \mathbf S_1^{(\Bc_{u}(|\Bc_{u}|))}
    \end{bmatrix}\mathbf B_1=\begin{bmatrix}
         \mathbf S_1^{(u)}\\
        \mathbf S_1^{(\Bc_{u}(1))}\\
        \vdots\\
         \mathbf S_1^{(\Bc_{u}(|\Bc_{u}|))}
    \end{bmatrix}\begin{bmatrix}
        \mathbf{0}\\
       \mathbf B
    \end{bmatrix}.\label{top-layer key information}
\end{align}

Since \(\mathbf B\) and \(\mathbf S_1\) are generated with elements uniformly i.i.d. over \(\mathbb F_\qsf\),  and 
\begin{align}
    (|\Bc_u|+1)\nsf_2
    \leq
    (\Usf-\rsf_1)\nsf_2
    =
    \Usf\nsf_2-\msf
    =
    \operatorname{rank}(\mathbf B),
\end{align}
the matrix in~\eqref{top-layer key information} has rank equal to $(|\Bc_{u}|+1)\nsf_2$, and these row vectors in~\eqref{top-layer key information} are linearly independent with high probability as \(\qsf\to\infty\). 
Moreover, the elements in $\mathbf{B}_{2,u}$ and $\{\mathbf B_{2,i}:i\in\Bc_u\}$ are generated independently with elements uniformly i.i.d. over $\mathbb F_\qsf$.

Hence, the matrix in Constraint~\ref{con: relay security} is row-wise full rank with high probability.

\section{Proofs of the information theoretic security against server}
\label{sec:proof of server security}
\subsection{Proof of the information theoretic security}
\label{sec:proof of the information theoretic security for server}
After receiving the messages from all relays, the server can recover
\begin{align} [\mathbf F_1\ \mathbf B_1]\mathbf W^\prime=
    \begin{bmatrix}
        \mathbf F & \mathbf 0\\
        \mathbf V & \mathbf B
    \end{bmatrix}
    \begin{bmatrix}
        \mathbf W\\
        \mathbf N
    \end{bmatrix}.
\end{align}
We denote 
\begin{align}
    \begin{bmatrix}
        T_1\\
        \vdots\\
        T_{\msf}\\
        T_{1,1}\\
        \vdots\\
        T_{1,\Usf\nsf_2-\msf}
    \end{bmatrix}= \begin{bmatrix}
        \mathbf F & \mathbf 0\\
        \mathbf V & \mathbf B
    \end{bmatrix}
    \begin{bmatrix}
        \mathbf W\\
        \mathbf N
    \end{bmatrix},
\end{align}
where $\{T_1,\ldots,T_\msf\}$  correspond to the desired aggregation task $ \mathbf F \mathbf W=\sum_{k\in[\Ksf]}W_k$.

Besides the relay-to-server transmissions,  the server can obtain key information through collusion. Recall that $\mathbf K_{\Tc}$ denotes the total key coefficient matrix revealed by the colluding users. Define
 \begin{align}
    \Delta_{\Tc}
    \triangleq
    \dim\left(
    \operatorname{rowspan}(\mathbf B)
    \cap
    \operatorname{rowspan}(\mathbf K_{\Tc})
    \right).
\end{align}
By Constraint~\ref{con: server security}, we have
\begin{align}
    \operatorname{rowspan}(\mathbf B)
    \cap
    \operatorname{rowspan}(\mathbf K_{\Tc})
    \subseteq
    \operatorname{rowspan}
    \left(
        \mathbf S_{1}^{\Bc}\mathbf B_1
    \right).
    \label{eq:server_security_constraint_intersection}
\end{align}
For each $u\in\Bc$, all datasets assigned to cluster $u$ are already covered by the colluding users. Hence, the corresponding data part $\mathbf S_{1}^{(u)}\mathbf F_1\mathbf W$ is contained in $W_{\Tc}$, and thus
\begin{align}
    H\left(
    \mathbf S_{1}^{\Bc}\mathbf F_1\mathbf W
    \middle|
    W_{\Tc}
    \right)=0.
    \label{eq:bad_cluster_data_known_server}
\end{align}

To prove the security constraint in~\eqref{eq: server security}, we have
\begin{subequations}
\begin{align}
&   I\Big(\{Y_u\}_{u\in[\Usf]};\{W_k\}_{k\in[\Ksf]}\Big|\sum_{k\in[\Ksf]}W_k,\{W_{\Tc},Z_{\Tc}\}\Big)\\
&\leq I\Big(
   \{T_1,\ldots, T_{\msf},T_{1,1},\ldots, T_{1,\Usf\nsf_2-\msf}\};
    \{W_k\}_{k\in[\Ksf]}
    \Big|\nonumber\\& \sum_{k\in[\Ksf]}W_k,\{W_{\Tc},Z_{\Tc}\}
    \Big)\label{T}\\
    &=I\Big(
   \{T_{1,1},\ldots, T_{1,\Usf\nsf_2-\msf}\};
    \{W_k\}_{k\in[\Ksf]}
    \Big|\nonumber\\& 
    \sum_{k\in[\Ksf]}W_k,\{W_{\Tc},Z_{\Tc}\}
    \Big)\label{eq2}\\
& =
H\Big(
    \{T_{1,1},\ldots, T_{1,\Usf\nsf_2-\msf}\}
    \Big|
   \sum_{k\in[\Ksf]}W_k,\{W_{\Tc},Z_{\Tc}\}
    \Big)\nonumber\\
&-
H\Big(
   \{T_{1,1},\ldots, T_{1,\Usf\nsf_2-\msf}\}
    \Big|
    \{W_k\}_{k\in[\Ksf]},\{W_{\Tc},Z_{\Tc}\}
    \Big)\\
& \leq
\left(\Usf\nsf_2-\msf-\Delta_{\Tc}\right)\Lsf^\prime
-
H\Big(
    \mathbf B\mathbf N
    \Big|
    Z_{\Tc}
    \Big)\label{eq:server_security_delta_bound}\\
& =
\left(\Usf\nsf_2-\msf-\Delta_{\Tc}\right)\Lsf^\prime
-
\left(\Usf\nsf_2-\msf-\Delta_{\Tc}\right)\Lsf^\prime\label{proof of server security}\\
&=0, 
\end{align}
\end{subequations}
where~\eqref{T} follows since
\(\{Y_u\}_{u\in[\Usf]}\) are linear combinations of
\(\{T_1,\ldots,T_{\msf},T_{1,1},\ldots,T_{1,\Usf\nsf_2-\msf}\}\), and hence
are determined by these recovered linear combinations.
\eqref{eq2} follows since
\(\{T_1,\ldots,T_{\msf}\}\) correspond to the desired aggregation task  and are
already given by \(\sum_{k\in[\Ksf]}W_k\). 
\eqref{eq:server_security_delta_bound} follows since, by the definition of $\Delta_{\Tc}$, the colluding users reveal $\Delta_{\Tc}$ key dimensions inside the top-layer key space $\operatorname{rowspan}(\mathbf B)$. By Constraint~\ref{con: server security}, the $\Delta_{\Tc}$ revealed key
dimensions inside $\operatorname{rowspan}(\mathbf B)$ lie in
\(\operatorname{rowspan}(\mathbf S_1^{\Bc}\mathbf B_1)\). Hence, there exists a row-wise full rank
coefficient matrix $\mathbf C_{\Tc}$ of dimension $\Delta_{\Tc}\times|\Bc|\nsf_2$. The key part
\(\mathbf C_{\Tc}\mathbf S_1^{\Bc}\mathbf B_1\mathbf N\) is known from
\(Z_{\Tc}\).  The data part
\(\mathbf C_{\Tc}\mathbf S_1^{\Bc}\mathbf F_1\mathbf W\) is known from
\(W_{\Tc}\) by \eqref{eq:bad_cluster_data_known_server}. Therefore, $\mathbf C_{\Tc}\mathbf S_1^{\Bc}
    [\mathbf F_1\ \mathbf B_1]\mathbf W^\prime$
is known from \(W_{\Tc}\) and \(Z_{\Tc}\), which gives
\(\Delta_{\Tc}\) independent linear combinations of
\([\mathbf F_1\ \mathbf B_1]\mathbf W^\prime\). Since the first $\msf$ rows correspond to the desired aggregation and are already conditioned on, this further determines $\Delta_{\Tc}$ independent linear combinations of the virtual vectors ${T_{1,1},\ldots,T_{1,\Usf\nsf_2-\msf}}$. Therefore, at most $\Usf\nsf_2-\msf-\Delta_{\Tc}$ dimensions of the virtual vectors remain unknown. \eqref{proof of server security} follows since the top-layer key space
\(\mathbf B\mathbf N\), which masks the virtual linear combinations
\({T_{1,1},\ldots,T_{1,\Usf\nsf_2-\msf}}\), contains
\(\Usf\nsf_2-\msf\) independent key dimensions, and the colluding key information \(Z_{\Tc}\triangleq\mathbf{K}_{\Tc}\mathbf{N}\) reveals exactly
\(\Delta_{\Tc}\) dimensions inside \(\operatorname{rowspan}(\mathbf B)\),
which gives \eqref{proof of server security}.

\subsection{Proof of Constraint~\ref{con: server security}} 
\label{proof of constraint 2}
We aim to prove that for each collection of colluding user sets \(\{\Tc_i:i\in[\Usf]\}\) satisfying  $|\Tc_i|\leq\Tsf_i$,
\begin{align}
    \operatorname{rowspan}(\mathbf B)
    \cap
    \operatorname{rowspan}(\mathbf K_{\Tc})
    \subseteq
    \operatorname{rowspan}
    \left(
        \mathbf S_1^{\Bc}\mathbf B_1
    \right).
\end{align}

Recall that
\begin{align}
    \mathbf B_1=
    \begin{bmatrix}
        \mathbf 0\\
        \mathbf B
    \end{bmatrix},
\end{align}
and for each cluster $i\in[\Usf]$,
\begin{align}
    \mathbf B_2^{(i)}
    =
    \begin{bmatrix}
        \mathbf S_1^{(i)}\mathbf B_1\\
        \mathbf B_{2,i}
    \end{bmatrix},
\end{align}
where $\mathbf B_{2,i}$ is generated independently with elements uniformly
i.i.d. over $\mathbb F_\qsf$.

Let $\Gc\triangleq[\Usf]\setminus\Bc$ denote the set of non-colluding
clusters.  
For each $i\in\Gc$, the revealed key coefficient matrix $\mathbf K_{\Tc}^{(i)}$ satisfies
\begin{align}
  \mathbf K_{\Tc}^{(i)}&=  \mathbf S_{2,i}^{\Tc_i}\mathbf B_2^{(i)}\nonumber\\
    &=
    \mathbf S_{2,i}^{\Tc_i}(.,[\nsf_2])
        \mathbf S_1^{(i)}\mathbf B_1
    +
    \mathbf S_{2,i}^{\Tc_i}(.,[\nsf_2+1:\Usf_i])
        \mathbf B_{2,i},
\end{align} and 
\begin{align}
    |\Tc_i|
    \leq \Vsf_i-\rsf_2^{(i)}
    \leq \Vsf_i-\ssf_2^{(i)}-\nsf_2
    =
    \Usf_i-\nsf_2,
\end{align}
where the second inequality follows from
$\nsf_2\leq \msf_2^{(i)}=\rsf_2^{(i)}-\ssf_2^{(i)}$.
Hence,
$\mathbf S_{2,i}^{\Tc_i}(.,[\nsf_2+1:\Usf_i])$
has at least $|\Tc_i|$ columns, and since $\mathbf S_{2,i}$ is generated
with elements uniformly i.i.d. over $\mathbb F_\qsf$, this submatrix has rank equal to $|\Tc_i|$ with high probability. Since $\mathbf B_{2,i}$ is generated with elements uniformly i.i.d. over $\mathbb F_\qsf$, all rows in $\mathbf K_{\Tc}^{(i)}$ are independent of
$\operatorname{rowspan}(\mathbf B)$ with high probability.

For each $i\in\Bc$, we consider the worst case where the colluding users reveal
the whole local key space of cluster $i$, which is
\begin{align}
  \mathbf K_{\Tc}^{(i)}=  \mathbf B_2^{(i)}
    =
    \begin{bmatrix}
        \mathbf S_1^{(i)}\mathbf B_1\\
        \mathbf B_{2,i}
    \end{bmatrix}.
\end{align}
Again, the rows of $\mathbf B_{2,i}$ are independent of
$\operatorname{rowspan}(\mathbf B)$ with high probability. Hence, the only
rows in the revealed key information that can overlap with
$\operatorname{rowspan}(\mathbf B)$ are the top-layer key rows
$\mathbf S_1^{(i)}\mathbf B_1$ for $i\in\Bc$. Therefore,
\begin{align}
    \operatorname{rowspan}(\mathbf B)
    \cap
    \operatorname{rowspan}(\mathbf K_{\Tc})
    \subseteq
    \operatorname{rowspan}
    \left(
        \mathbf S_1^{\Bc}\mathbf B_1
    \right),
\end{align}
which proves Constraint~\ref{con: server security}.

\end{document}